\documentclass[a4paper,twoside,12pt]{article}
\usepackage[english]{babel} 
\usepackage[latin1]{inputenc}
\usepackage{cmbright}
\usepackage[T1]{fontenc}
\usepackage{lmodern}
 
\usepackage{cite}

\usepackage{maplestd2e}

\usepackage{xcolor}
\usepackage{hyperref}
\definecolor{darkgreen}{rgb}{0,0.4,0}
\definecolor{BrickRed}{rgb}{0.65,0.08,0}
\hypersetup{colorlinks=true,linkcolor=blue,citecolor=red,filecolor=BrickRed,urlcolor=darkgreen}
\usepackage{subfig}
\usepackage{graphicx} 
\graphicspath{{./}{../../pics/}}
\makeatletter
\def\input@path{{./}{../../pics/}}  
\makeatother

\usepackage{tikz}
\usetikzlibrary{calc}

\usepackage{verbatim} 
\usepackage{amsmath,amssymb,amsthm}
\usepackage{color}
\usepackage{cite}
\usepackage{enumerate}
\usepackage{array}
\usepackage{footnote}
\usepackage[math]{cellspace}
\usepackage{bbding} 
\usepackage{a4wide} 
\newcommand{\exzk}[2]{[z^{#1}]_{#2}}
\newcommand{\exzkn}{\exzk{n}{\zeta_k}}
\newcommand{\coefb}{C_7}
\newcommand{\coefd}{C'_7}
\newcommand{\walksym}{\omega}
\newcommand{\walk}[1]{\walksym_{#1}}

\newcommand{\stepset}{\mathcal{S}}

\newcommand{\LandauO}{\mathcal{O}}

\newcommand{\Dc}{\mathcal{D}}

\newcommand{\C}{\mathbb{C}}
\newcommand{\N}{\mathbb{N}}
\newcommand{\Q}{\mathbb{Q}}
\newcommand{\Z}{{\mathbb Z}}

\newtheorem{theo}{Theorem}[section]
\newtheorem{lemma}[theo]{Lemma}
\newtheorem{prop}[theo]{Proposition}

\newtheorem{definition}[theo]{Definition}
\newenvironment{example}[1][]{\refstepcounter{theo} \medskip \noindent \textbf{\textit{Example \thetheo #1:}} }{ \hfill $_{\blacksquare} $\\}
\newenvironment{examplenoend}[1][]{\refstepcounter{theo} \medskip \noindent \textbf{\textit{Example \thetheo #1:}} }{}

\providecommand{\keywords}[1]{\textbf{\textit{Keywords: }} #1}

\usepackage{fancyhdr}
\usepackage{titling}
\begin{document}
\baselineskip \the\baselineskip plus .1pt
\author{Cyril Banderier  \and  Michael Wallner}
\title{The kernel method for\\ lattice paths below a line of rational slope}  
\date{}
\maketitle

\begin{center}
\noindent \url{http://lipn.fr/~banderier/},  CNRS \& Univ. Paris Nord, France \\
\url{http://dmg.tuwien.ac.at/mwallner/}, TU Wien, Austria
\end{center}

 \vspace{3mm} 
\begin{tabular}{cc}
\begin{tabular}{c} \includegraphics[width=3cm]{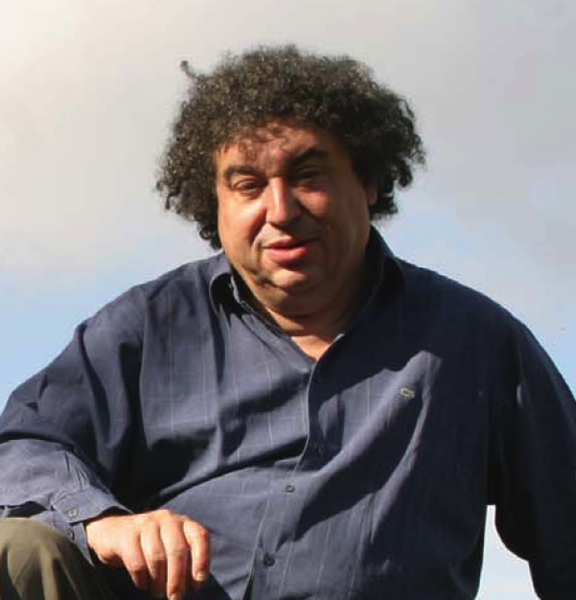} \end{tabular} &
\begin{tabular} {c}
\SixFlowerAltPetal\SixFlowerAltPetal\SixFlowerAltPetal\\
 \multicolumn{1}{p{11cm}}{{\em  \quad We dedicate this article to the memory of Philippe Flajolet, who was and will remain a guide and a wonderful source of inspiration for so many of us. \quad}} \\
 \SixFlowerAltPetal\SixFlowerAltPetal\SixFlowerAltPetal
\end{tabular}
\end{tabular}
 \, \vspace{3mm}

{\it \small
This article corresponds to the article accepted for publication in the Developments in Mathematics Series (Springer), associated with the 8th International Conference on Lattice Path Combinatorics and Applications.
This work considerably extends our preliminary version ``Lattice paths of slope~2/5'' which appeared in 
the Proceedings of the ANALCO15 San Diego Conference.}

\begin{abstract}
We analyse some enumerative and asymptotic properties of lattice paths below a line of rational slope.
We illustrate our approach with Dyck paths under a line of slope~$2/5$.
This answers Knuth's problem \#4 from his ``Flajolet lecture'' during the conference ``Analysis of Algorithms'' (AofA'2014) in Paris in June 2014.
Our approach extends the work of Banderier and Flajolet for asymptotics and enumeration of directed lattice paths to 
the case of generating functions involving several dominant singularities, and has applications 
to a full class of problems involving some ``periodicities''.

A key ingredient in the proof is the generalization of an old trick by Knuth himself 
(for enumerating permutations sortable by a stack),
promoted by Flajolet and others as the ``kernel method''.
All the corresponding generating functions are algebraic,
and they offer some new combinatorial identities, 
which can also be tackled in the {\em A=B} spirit of Wilf--Zeilberger--Petkov{\v s}ek.

We show how to obtain similar results for any rational slope. 
An interesting case is e.g.~Dyck paths below the slope~$2/3$
(this corresponds to the so-called Duchon's club model), for which we solve a conjecture
related to the asymptotics of the area below such lattice paths.
Our work also gives access to 
lattice paths below an irrational slope (e.g.~Dyck paths below $y=x/\sqrt{2}$),
a problem that we study in a companion article.
\end{abstract}

\pagebreak
\section{Introduction}
\label{sec:intro}

For the enumeration of simple lattice paths (allowing just the jumps $-1$, $0$, and $+1$), many methods are often used, 
like e.g.~the Lagrange inversion, determinant techniques, continued fractions, orthogonal polynomials, bijective proofs, and a lot is known in such cases~\cite{Flajolet80,Kr15,Mohanty79,Narayana79}. 
These nice methods do not apply to more complex cases of more generic jumps (or, if one adds a spacial boundary, like a line of rational slope).
It is then possible to use some ad hoc factorization due to Gessel~\cite{Gessel80}, 
or context-free grammars to enumerate such lattice paths~\cite{LabelleYeh90,Merlini96,Duchon00}.
One drawback of the grammar approach is that it leads to heavy case-by-case computations (resultants of equations of huge degree).
In this article, we show how to proceed for the enumeration and the asymptotics in these harder cases:
our techniques are relying on the ``kernel method'' 
which (contrary to the context-free grammar approach) offers access to the true simple {\em generic} structure of the final generating functions 
and the {\em universality} of their asymptotics via singularity analysis. 

Let us start with the history of what Philippe Flajolet named  the ``kernel method'': 
It has been part of the folklore of combinatorialists for some time and its simplest application deals with functional equations 
(with apparently more unknowns than equations!) of the form
\[
K(z,u)F(z,u)= p(z,u)+q(z,u) G(z),
\]
where the functions $p, q$, and $K$ are given and where $F, G$ are the unknown generating functions we want to determine.
$K(z,u)$ is a polynomial in $u$ which we call the ``kernel'' as
we ``test'' this functional equation on functions $u(z)$ canceling this kernel\footnote{The ``kernel method'' that we mention here for functional equations in combinatorics has nothing to do with what is known as the ``kernel method'' or ``kernel trick'' in statistics or  machine learning. Also, there is no integral directly related to our kernel. For sure, in our case the word kernel was chosen as its zeros  will play a key role, and also, in one sense,  as this kernel has in its core the full description of the problem, and its resolution.}.
The simplest case is when there is only one branch, $u_1(z)$, such that $K(z,u_1(z))=0$ and $u_1(0)=0$;
in that case, a single substitution gives a closed-form solution for $G$:
namely, $G(z)=-p(z,u_1(z))/q(z,u_1(z))$.

Such an approach was introduced in 1969 by Knuth to enumerate permutations sortable by a stack, 
see the detailed solution to Exercise 2.2.1--4 in {\em The Art of Computer Programming}
(\!\cite[pp.~536--537]{Kn69} and also Ex.~2.2.1.11 therein),
which presents a ``new method for solving the ballot problem'', 
for which the kernel $K$  is a quadratic polynomial (this specific case involves just one branch~$u_1(z)$).

In combinatorics exist many applications of this method for solving variants of the above functional equation:
one is known as the ``quadratic method'' in map enumeration,
as initially developed in 1965 by Brown during his collaboration with Tutte (see Section 2.9.1 from~\cite{Brown65}, 
and~\cite{bfss01} for the analysis of about a dozen families of maps).  
During nearly 30 years, the kernel method was dealing only with ``quadratic cases''  like the ones of Brown for maps or 
of Knuth for a vast amount of examples involving trees, polyominoes, walks~\cite{Prodinger03},    
or more exotic applications like e.g.~the one mentioned by Odlyzko in his wonderful survey on asymptotic methods in enumeration~\cite{FanGrahamOdlyzko95}.
Then, in 1998, the initial approach by Knuth was generalized by a group of four people, all of them being in contact and benefiting from mutual insights: 
Banderier in his memoir~\cite{Ba98} solved some problems related to generating trees and walks, this later lead to the article with Flajolet~\cite{BaFl02} 
and to the solution of some conjectures due to Pinzani in the article with Bousquet-M\'elou et al.~\cite{hexa}.
At the same time, Petkov{\v s}ek analysed linear multivariate recurrences in~\cite{Pet98}, a work later extended in~\cite{BoPe00}.
All these articles contributed to turn the original approach by Knuth into a method working when the equation has more unknowns (and the kernel has more roots).
This solves equations of the type
\[
K(z,u)F(z,u)= \sum_{i=1}^m p_i(z,u) G_i(z),
\]
where $K$ and the $p_i$'s are known polynomials, and where $F$ and the $G_i$'s are unknown functions.

A few years later, Bousquet-M\'elou and Jehanne~\cite{BoJe06} solved the case of algebraic equations in $F$ of arbitrary degree:
$$P(z,u,F(z,u),G_1(z),\dots,G_m(z))=0.$$

The kernel method thus plays a key role in many combinatorial problems. A few examples are directed lattice paths and their asymptotics~\cite{BaFl02, BousquetMelou08}, additive parameters like area~\cite{BaGi06, Schwerdtfeger14}, 
generating trees~\cite{hexa}, pattern avoiding permutations~\cite{Mansour}, prudent walks~\cite{Duchi05,BacherBeaton14}, 
urn models~\cite{SchulteGeers15}, statistics in posets~\cite{Fusy}, and many other nice combinatorial structures...

\smallskip
Independently, in probability theory, in the '70s, Malyshev developed an approach now sometimes called the ``iterated kernel method'' in order to analyse
nearest neighbour random walks in queuing theory. These lead to the following type of equations:
\begin{align*}
	K(t,x,y) F(t,x,y) = p_0(t,x,y)  + p_1(t,x,y)  F(x,0) + p_2(t,x,y)  F(0,y),
\end{align*}
where $K$ and the $p_i$'s are known polynomials, while $F$ is the unknown  function we are looking for. This approach culminated in the book~\cite{Fayolle99}, which
was later revisited in the 2000s (e.g. in~\cite{raschel}), also with a more combinatorial point of view in~\cite{BousquetMelouMishna10}. 
It is still the subject of vivid activities, including the extension to higher dimensions~\cite{BoBoKaMe14}.
Moreover, the kernel method also gives the transient solution of some birth-death queuing processes~\cite{Mohanty07}.

\smallskip
Also independently, in statistical mechanics, several authors developed other incarnations of the kernel method.
 E.g., the WKB limit of the Bethe Ansatz (also called thermodynamical Bethe Ansatz) often leads to algebraic equations and to what is called the algebraic Bethe Ansatz~\cite{Gaudin83}.
The kernel method is also used in the study of the Ising model of bicoloured maps 
(see Theorem~8.4.5 in~\cite{Eynard16}, and pushing further this method led Eynard to his ``topological recurrence''),
and in many articles on enumeration related to directed animals, polymers, walks~\cite{vanRensburgRechnitzer01, vanRensburg05, vanRensburgPrellbergRechnitzer08}.

\bigskip

\pagebreak
After this short history of the kernel method, we want to show how to use it to derive explicit counting formulae and asymptotics for directed lattice paths below a line of rational slope.
In the article by Banderier \& Flajolet~\cite{BaFl02}, the class of directed lattice paths in $\Z^2$ was investigated thoroughly by means of analytic combinatorics (see \cite{flaj09}).
Our work is an extension of this article in mainly five ways:
\begin{itemize}
\item[1.] Our work involves lattice paths having a ``periodic support'',
the comment in~\cite[Section~3.3]{BaFl02} was incomplete for this more cumbersome case,
indeed there are then several dominant singularities, and we had to revisit in more detail the structural properties of the roots associated to the kernel method
in order  to understand the contribution of each of these singularities.
It is pleasant that this new understanding gives a tool 
to deal with the asymptotics of many other lattice path enumeration problems.
\item[2.] We get new explicit formulae for the generating functions of walks with starting and ending at altitude other than 0, and links with complete symmetric homogeneous polynomials.
\item[3.] We give new closed forms for the coefficients of these generating functions.
\item[4.] We have an application to some harder parameters (like the area below a lattice path).
\item[5.] We extend the results to walks below a line of \textit{arbitrary rational} slope, 
             paving the way for our forthcoming article on walks below a line of arbitrary \textit{irrational} slope~\cite{BanderierWallner16}.
\end{itemize}

Let us give a definition of the lattice paths we consider:
 \begin{definition}[Jumps and lattice paths] \label{def:LP}
 A \emph{step set} $\stepset \subset \Z^2$ is a finite set of vectors $\{ (x_1,y_1), \ldots, (x_m,y_m)\}$. 
An $n$-step \emph{lattice path} or \emph{walk} is a sequence of vectors $(v_1,\ldots,v_n)$, such that $v_j$ is in $\stepset$. 
Geometrically, it may be interpreted as a sequence of points $\walksym =(\walk{0},\walk{1},\ldots,\walk{n})$ where $\walk{i} \in \Z^2,~\walk{0} = (0,0)$ (or another starting point)
and $\walk{i}-\walk{i-1} = v_i$ for $i=1,\ldots,n$.
The elements of $\stepset$ are called \emph{steps} or \emph{jumps}. 
The \emph{length} $|\walksym|$ of a lattice path is its number $n$ of jumps. 
 \end{definition}
The lattice paths can have different additional constraints shown in Table~\ref{fig-4types}.
  \begin{table*}[t]
 \small
 \begin{center}\renewcommand{\tabcolsep}{3pt}
 \begin{tabular}{|c|c|c|}
 \hline
 & ending anywhere & ending at 0\\
 \hline
 \begin{tabular}{c} unconstrained \\ (on~$\Z$) \end{tabular}
 & \begin{tabular}{c} 

 {\includegraphics[width=63mm]{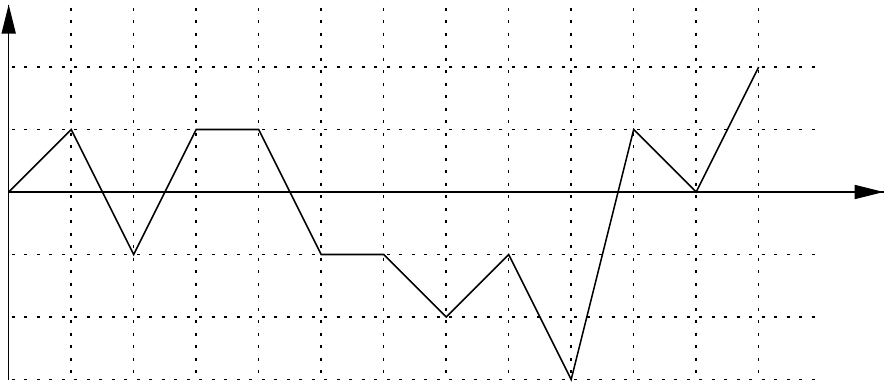}} 
\\ 
 walk/path ($\cal W$) 
 \end{tabular}
 & \begin{tabular}{c} 

 {\includegraphics[width=63mm]{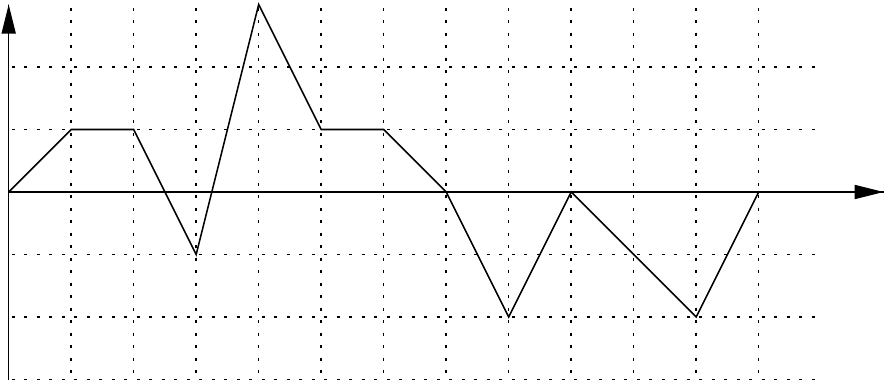}} 
\\
 bridge ($\cal B$)
 \end{tabular} \\
 \hline
 \begin{tabular}{c}constrained\\ (on $\N$) \end{tabular}
 & \begin{tabular}{c} 
 \includegraphics[width=63mm]{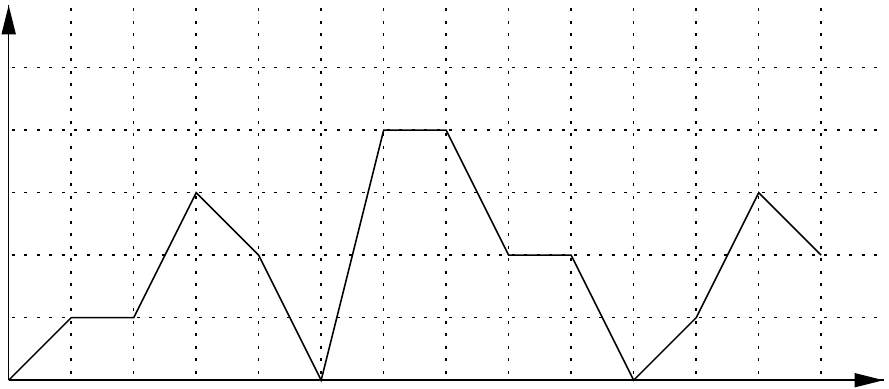} 
\\ 
 meander ($\cal M$)\\ 
 \end{tabular}
 & \begin{tabular}{c} 
 {\includegraphics[width=63mm]{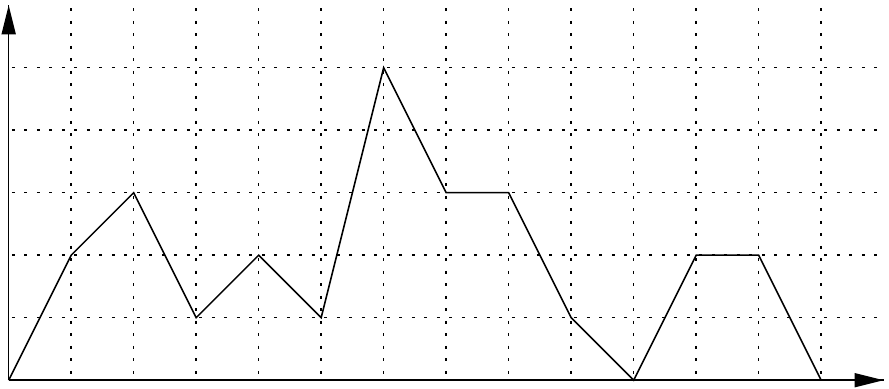}} 
\\ 
 excursion ($\cal E$)\\ 
 \end{tabular}\\
 \hline
 \end{tabular}
 \end{center}
 \caption{\label{fig-4types} 
 The four types of paths: walks, bridges, meanders, and excursions.
We refer to these walks as the Banderier--Flajolet model,  in contrast to the model in which we will consider lattice paths below a rational slope boundary.}
 \end{table*}
 
We restrict our attention to \emph{directed paths} which are defined by the fact that, for each jump $(x,y) \in \stepset$, one must have $x \geq  0$.  
The next definition allows to merge the probabilistic point of view (random walks) and the combinatorial point of view (lattice paths):
 \begin{definition}[Weighted lattice paths]
 For a given step set $\stepset = \{s_1,\ldots,s_m\}$, we define the respective \emph{system of weights}\index{lattice path!weights} 
as $ \{w_1,\ldots,w_m\}$ where $w_j >0$ is the weight associated to step $s_j$ for $j=1,\ldots,m$. 
The \emph{weight of a path} is defined as the product of the weights of its individual steps. 
 \end{definition}

\pagebreak
\textbf{Plan of this article.} 
\begin{itemize}
\item 
First, in Section~\ref{sec:imaginary}, we recall the fundamental results for lattice paths below a line of slope~$\alpha$ (where $\alpha$ is an integer or the inverse of an integer), and the links with trees.
\item 
Then, in Section~\ref{sec:KnuthProblem}, we give Knuth's open problem on lattice paths below a line of slope~$2/5$. 
\item 
In Section~\ref{sec:bijection}, we give a bijection between lattice paths below any line of rational slope,
and lattice paths from the Banderier--Flajolet model.
\item 
In Section~\ref{sec:funceq}, the needed bivariate generating function is defined and the governing functional equation is derived and solved:
here the ``kernel method'' plays the most significant role in order to obtain the generating function 
(as typical for many combinatorial objects which are recursively defined with a ``catalytic parameter''). 
\item 
In Section~\ref{sec:asymptotics}, we tackle some questions on asymptotics, thus answering the question of Knuth.
\item 
In Section~\ref{sec:nakatoku}, we comment on links with previous results of Nakamigawa and Tokushige, which motivated Knuth's problem, and we explain why some cases lead to particularly striking new closed-form formulae.
\item 
In Section~\ref{sec:duchon}, we analyse what happens for the Duchon's club model (lattice paths below a line of slope~$2/3$), and we extend our formulae to general rational slopes.
\end{itemize}

\pagebreak

\section[Trees, frac.~trees, imag.~trees]{Trees, fractional trees, imaginary trees}
\label{sec:imaginary}

Due to their fundamental role in computer science trees were the subject of many investigations,
and there exist many alternative representations of this key data structure. One of the most useful ones
is an encoding by ``traversing'' the  tree via a depth-first traversal (or via a breadth-first traversal). This directly gives a lattice path 
associated to the original tree.
In fact,  what are called ``simple families of ordered trees'' 
(rooted ordered trees  in which each node has a degree prescribed to be in a given set) are in bijection 
with lattice paths. The reason is the famous \emph{{\L}ukasiewicz correspondence} between trees and lattice paths, see Figure~\ref{fig:LukaBijection}.

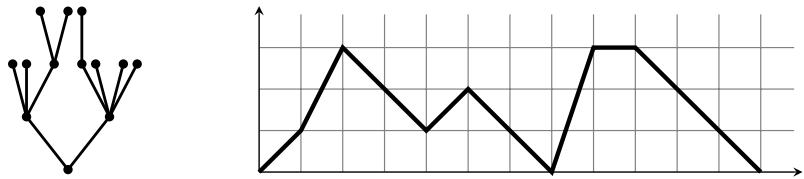
\begin{figure}[!ht] 
	\centering
	\scalebox{0.7}{\newcommand{\scaleA}{0.26}      
\newcommand{\scaleB}{0.26} 
\newcommand{\scaleC}{0.26} 

\newcommand{\pathlinewidth}{1.5}
\newcommand{\gridlinewidth}{thin}

\tikzset{dot/.style={circle,fill=black,inner sep=0,minimum size=5pt}}
                     
\begin{tikzpicture}[>=stealth]
			
	\draw[-, line width=\pathlinewidth] (-1*\scaleC,0) -- (-4*\scaleA,1){};
	\draw[-, line width=\pathlinewidth] (-1*\scaleC,0) -- (2*\scaleA,1){};
	
	\draw[-, line width=\pathlinewidth] (-4*\scaleA,1) -- (-5*\scaleB,2){};
	\draw[-, line width=\pathlinewidth] (-4*\scaleA,1) -- (-4*\scaleB,2){};	
	\draw[-, line width=\pathlinewidth] (-4*\scaleA,1) -- (-2*\scaleB,2){};
	
	\draw[-, line width=\pathlinewidth] (-2*\scaleB,2) -- (-3*\scaleC,3){};
	\draw[-, line width=\pathlinewidth] (-2*\scaleB,2) -- (-1*\scaleC,3){};	
		
	\draw[-, line width=\pathlinewidth] (2*\scaleA,1) -- (0*\scaleB,2){};
	\draw[-, line width=\pathlinewidth] (2*\scaleA,1) -- (1*\scaleB,2){};	
	\draw[-, line width=\pathlinewidth] (2*\scaleA,1) -- (3*\scaleB,2){};
	\draw[-, line width=\pathlinewidth] (2*\scaleA,1) -- (4*\scaleB,2){};
	
	\draw[-, line width=\pathlinewidth] (0*\scaleB,2) -- (0*\scaleC,3){};
	
	\node [dot] at  (-1*\scaleC,0) {};
	\node [dot] at  (-4*\scaleA,1) {};
	\node [dot] at  (2*\scaleA,1) {};
	\node [dot] at  (-5*\scaleB,2) {};
	\node [dot] at  (-4*\scaleB,2) {};
	\node [dot] at  (-2*\scaleB,2) {};
	\node [dot] at  (-3*\scaleC,3) {};
	\node [dot] at  (-1*\scaleC,3) {};
	\node [dot] at  (0*\scaleB,2) {};
	\node [dot] at  (1*\scaleB,2) {};
	\node [dot] at  (3*\scaleB,2) {};
	\node [dot] at  (4*\scaleB,2) {};
	\node [dot] at  (0*\scaleC,3) {};
	
\end{tikzpicture}} \qquad
	\scalebox{1.1}{\newcommand{\scale}{0.5}      
\newcommand{\pathlinewidth}{1.5}
\newcommand{\gridlinewidth}{thin}     
                
\begin{tikzpicture}[>=stealth]
	
	\draw[step=.5cm,gray,\gridlinewidth] (0,0) grid (6.4,1.9);
	
	\draw[->] (0,0) -- (0,2) {};
	\draw[->] (0,0) -- (6.5,0) {};
	
	\draw[-, line width=\pathlinewidth] (0,0)
	\foreach \x/\y in  {1/1, 2/3, 3/2, 4/1, 5/2, 6/1, 7/0, 8/3, 9/3, 10/2, 11/1, 12/0}
	{	 -- (\scale*\x,\scale*\y) };
	
\end{tikzpicture}}
	\caption{The {\L}ukasiewicz bijection between trees and lattice paths:  A little fly is traveling along the full contour of the tree 
	starting from the root. Whenever it meets a new node, one draws a new jump of size ``arity of the node $-1$'' in the lattice path. Without loss of generality, one can always remove the very last jump (as it will always be a ``$-1$'') and thus we get an excursion which is in bijection with the initial tree. 
It is straightforward to reverse this bijection. Additionally, note that any deterministic traversal of the tree offers such a bijection, so it could be a depth-first traversal, but also {e.g.}~a breadth-first traversal.} 
	\label{fig:LukaBijection}
\end{figure}

Basic manipulations on lattice paths also show that \emph{Dyck paths} (paths with jumps North and East, see Figure~\ref{traversal})
below the line $y= \alpha x$
 ($\alpha$ being here a positive integer), or below the line $y= x/\alpha $, are in bijection with trees (of arity $\alpha$, i.e.,~every node has exactly $0$ or $\alpha$ children). 
 
\begin{figure}[!hb]
\scalebox{0.95}{
\begin{tabular}{ccc}
       \begin{tabular}{c}
	\scalebox{0.92}{\newcommand{\scaleA}{0.26}      
\newcommand{\scaleB}{0.26} 
\newcommand{\scaleC}{0.26} 

\newcommand{\pathlinewidth}{1.5}
\newcommand{\gridlinewidth}{thin}

\tikzset{dot/.style={circle,fill=black,inner sep=0,minimum size=5pt}}
                     
\begin{tikzpicture}[>=stealth]
			
	\draw[-, line width=\pathlinewidth] (0*\scaleC,0) -- (-6*\scaleA,1){};
	\draw[-, line width=\pathlinewidth] (0*\scaleC,0) -- (0*\scaleA,1){};
	\draw[-, line width=\pathlinewidth] (0*\scaleC,0) -- (6*\scaleA,1){};
	
	\draw[-, line width=\pathlinewidth] (0*\scaleA,1) -- (-2*\scaleB,2){};
	\draw[-, line width=\pathlinewidth] (0*\scaleA,1) -- (0*\scaleB,2){};	
	\draw[-, line width=\pathlinewidth] (0*\scaleA,1) -- (2*\scaleB,2){};
	
	\draw[-, line width=\pathlinewidth] (2*\scaleB,2) -- (0*\scaleC,3){};
	\draw[-, line width=\pathlinewidth] (2*\scaleB,2) -- (2*\scaleC,3){};	
	\draw[-, line width=\pathlinewidth] (2*\scaleB,2) -- (4*\scaleC,3){};	
		
	\draw[-, line width=\pathlinewidth] (6*\scaleA,1) -- (4*\scaleB,2){};
	\draw[-, line width=\pathlinewidth] (6*\scaleA,1) -- (6*\scaleB,2){};	
	\draw[-, line width=\pathlinewidth] (6*\scaleA,1) -- (8*\scaleB,2){};
	
	\node [dot] at  (0*\scaleC,0) {};
	\node [dot] at  (-6*\scaleA,1) {};
	\node [dot] at  (0*\scaleA,1) {};
	\node [dot] at  (6*\scaleA,1) {};
	\node [dot] at  (-2*\scaleB,2) {};
	\node [dot] at  (0*\scaleB,2) {};
	\node [dot] at  (2*\scaleB,2) {};
	\node [dot] at  (0*\scaleC,3) {};
	\node [dot] at  (2*\scaleC,3) {};
	\node [dot] at  (4*\scaleC,3) {};
	\node [dot] at  (4*\scaleB,2) {};
	\node [dot] at  (6*\scaleB,2) {};
	\node [dot] at  (8*\scaleB,2) {};
	
\end{tikzpicture}} 
       \\ \qquad \\
       	\includegraphics[width=60mm]{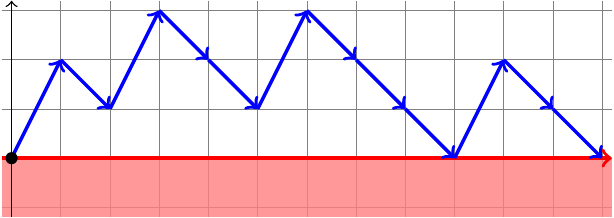} 
	\end{tabular} &
	\begin{tabular}{c}
	\includegraphics[width=30mm]{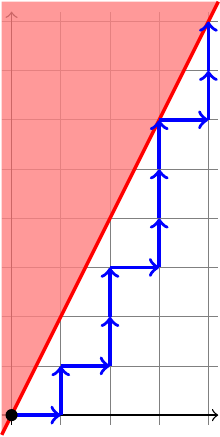} 
	\end{tabular}& 
	\begin{tabular}{c}
	\includegraphics[width=50mm]{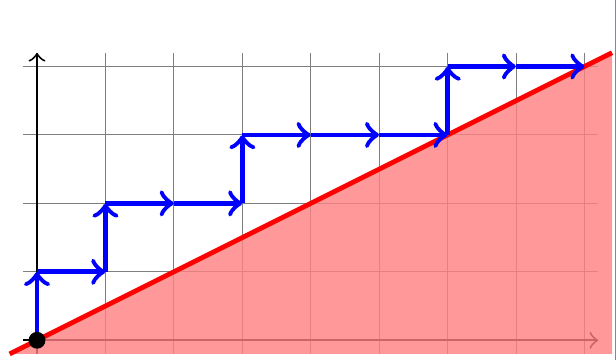} \\ \qquad \\
	\includegraphics[width=50mm]{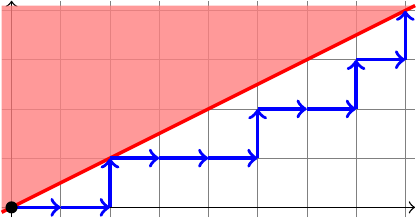} 
	\end{tabular} 
       \end{tabular}
       }
	\caption{Examples of combinatorial structures which are in bijection: ternary trees, excursions of directed lattice paths with jumps $+2$ and $-1$, 
	Dyck paths of North-East steps below the line $y = 2x$, Dyck paths above the line $y=\frac{1}{2} x$, and 
	Dyck paths below the line $y=\frac{1}{2}x$.
	\label{traversal}}
\end{figure}
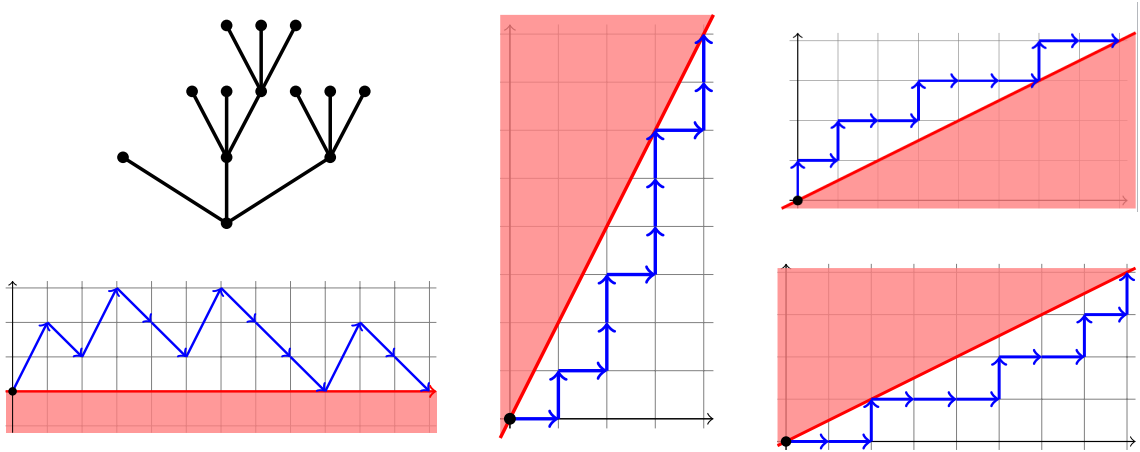

\pagebreak

The generating function $F(z)=\sum f_n z^n$, where $f_n$ counts the number of trees with $n$ nodes (internal and external ones) satisfies the functional equation
 $F(z) =z \phi( F(z))\,,$
 where $\phi$ encodes the allowed arities. Thus, we get 
binary trees: $\phi(F)=1+F^2$, unary-binary trees: $\phi(F)=1+F+F^2$, 
$t$-ary trees:  $\phi(F)=1+F^t$, general trees: $\phi(F)=1/(1-F)$. 
See~\cite{flaj09} for more on this approach, also extendible to unordered trees (i.e.,~the order of the children is not taken into account).

Because of the bijection with lattice paths, the enumeration of ordered trees solves the question of lattice paths below a line of integer slope.
In the simplest case of classical Dyck paths, many tools were developed.    
In 1886, Delannoy was the first to promote a systematic way to enumerate lattice paths,
using recurrences and an array representation (see~\cite{BanderierSchwer05} for more on this). 
Then, the Bertrand ballot problem~\cite{Bertrand87} (already previously considered by Whitworth) and the ruin problem (as studied along centuries by Fermat, Pascal, the Bernoullis,  Huygens, de Moivre,  Lagrange, Laplace, Amp\`ere and Rouch\'e)
 were a strong motor for the birth of the combinatorics of lattice paths,
 one famous solution being the one by Andr\'e~\cite{Andre87} via a bijective proof  involving ``good minus bad'' paths. 
Aebly~\cite{Aebly23} and Mirimanoff~\cite{Mirimanoff23} gave a geometric variant of this bijective proof, which corresponds to what is nowadays known as the reflection principle.
Later, the cycle lemma by Dvoretsky and Motzkin~\cite{DvoretzkyMotzkin47}  proved useful for many similar problems.
During the last century, all these tools were extended and applied to other cases than the classical Dyck paths, and we will use some of them in this article.

With respect to the closed form for the enumeration, another powerful tool
is the Lagrange--B\"urmann inversion formula (see e.g.~\cite{flaj09}). Applied on  $T(z) = 1 + z T(z)^t$ (the equation for the generating function of $t$-ary trees where $z$ marks internal nodes), it gives
\begin{equation} T(z)^r= \sum_{k\geq 0} \binom{tk+r}{k} \frac{r}{tk+r} z^k = \sum_{k\geq 0} \binom{tk+(r-1)}{k} \frac{r}{(t-1)k+r} z^k  \,.
\end{equation}

\vspace{-2mm}
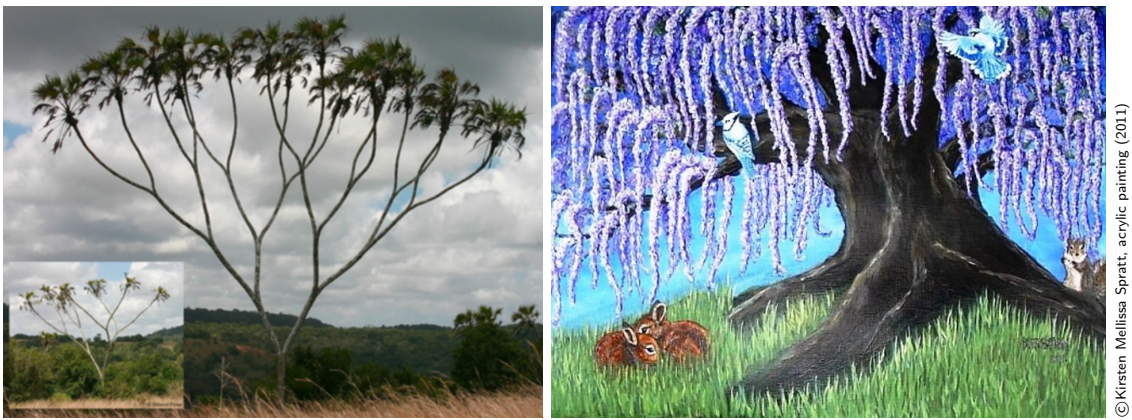
\begin{figure}[!hb] 
\centering
\scalebox{0.90}{
\includegraphics[width=0.4865\textwidth]{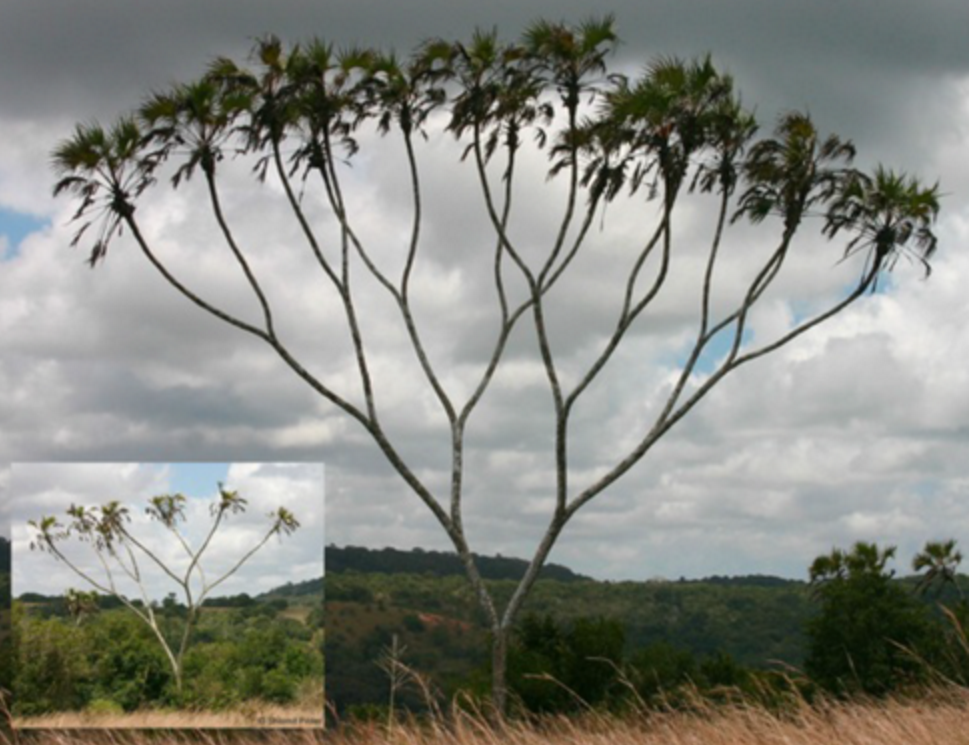}
\includegraphics[width=0.5\textwidth]{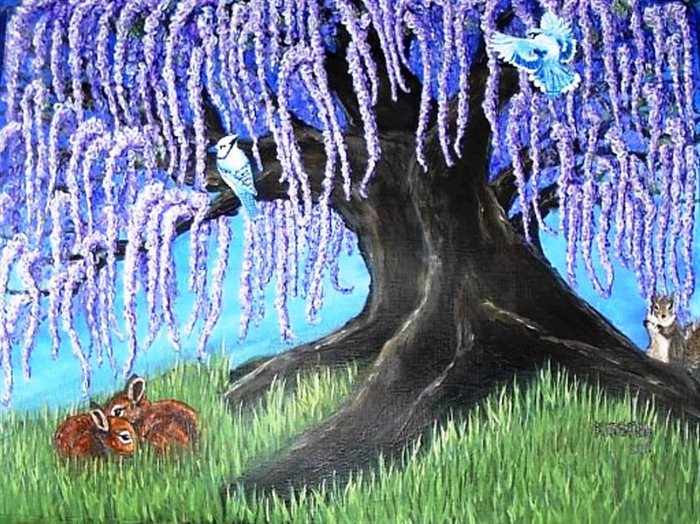}
\begin{picture}(50,50)
  \put(0,0){\rotatebox{90}{\tiny 
\copyright Kirsten Mellissa Spratt, acrylic painting (2011)}}
\end{picture}
}

\caption{It is possible to plug any value for $t$ in $T(z)$, which is known to count trees and lattice paths when $t$ is an integer. 
What happens when we consider generalized binomial series of order $3/2$, or of other fractional values? 
To recycle a nice pun by Don Knuth~\cite{knuthtalk}: Nature is offering nice binary trees, will imaginary trees one day play a role in computer science? \label{itree}}
\end{figure}

Plugging rational values is not directly leading to a power series with integer coefficients, but it ``miraculously'' becomes the case after basic transformations. For example,
as observed by Knuth~\cite{knuthtalk}, for $t=3/2$, one has the following neat non-trivial identity:

\begin{equation}
T(z) T(-z) = \left( \sum_{k\geq 0} \frac{\binom{3k/2}{k}}{k/2+1} z^k   \right) \left(\sum_{k\geq 0} \frac{\binom{3k/2}{k}}{k/2+1} (-z)^k \right) 
= \sum_{n \geq 0} \frac{\binom{3n+1}{n} }{n+1} z^{2n}\,.
\end{equation}

What could be the meaning of such identities involving ``half-trees''? The explanation behind this formula is better
seen in terms of lattice paths, and we will shed light on it in the next sections via the kernel method.
Another set of mysterious identities is e.g.~incarnated by:
\begin{equation}
\ln T(z) = \ln \sum_{n\geq 0} \frac{\binom{tn}{n}}{(t-1)n+1} z^n = \sum_{n\geq 1} \frac{\binom{tn}{n}}{tn} z^n\,.
\end{equation}

In fact, this one is just another avatar of the cycle lemma, which is also the reason 
for the link between the generating function of bridges and the generating function of excursions 
(a fact also appearing in various disguises {e.g.}~in the  Spitzer formula, in the Sparre Andersen formula), 
see~\cite{BaFl02} for explanations and proofs.

As we have seen, Dyck paths below an integer slope (or structures in bijection with them) were subject to many approaches, now considered as ``folklore''.
The first result for lattice paths below a rational slope came much later, and is best summarized by the following theorem:
\begin{theo}[Bizley's formula, Grossman's formula]
The number $f(an,bn)$ of  Dyck paths from $(0,0)$ to $(a n ,bn )$ staying weakly  above $y=\frac{a}{b}x$ is given by the following expressions, where $c_j := \frac{1}{a j+ b j }\binom{ a j+ b j}{a j}$:
\begin{eqnarray} \label{Gf} 
f(an,bn) =&\displaystyle  [t^n] \exp \sum_{j\geq 0}^n \frac{1}{(a+b)} \binom{(a+b) j}{a} t^j\,,\\ \label{coeff}
f(an,bn)  =&  \displaystyle \sum_{\left\{\substack{\text{integer partitions of } n:\\  \sum_{j=1}^{k}  j \, e_j = n }\right\}}   \,  \prod_{j=1}^{k}  \frac{(c_{j})^{e_j}}{e_j!} \,.
\end{eqnarray}
\end{theo}

Formula~\eqref{coeff} was first stated without proof by Grossman in 1950. A proof was then given by Bizley~\cite{Bizley54} in 1954.
It starts with Formula~\eqref{Gf}, which is an avatar of the cycle lemma~\cite{DvoretzkyMotzkin47} expressed in terms of a generating function.
Then routine power series manipulation gives Formula~\eqref{coeff}.
These formulae (or special cases of them) have since been rediscovered (and published...) many times.
One nice modern formulation of the method behind is found in the article by Gessel~\cite{Gessel80}. 
There exist alternative generic formulae as given by Banderier and Flajolet~\cite{BaFl02}, Sato~\cite{Sato89}, which simplify for ad hoc cases~\cite{Duchon00, KKKK16}.

 This formula admits many extensions as one could for example add parameters or take into account certain patterns.
This would lead to ``rational'' Narayana numbers, ``rational'' q-analogs, ``rational'' Mahonian statistics (on lattice paths!), etc.

 For each $n$, Grossman's formula~\eqref{coeff} for $f(an,bn)$ involves $p(n)$ summands, where $p(n)$ is the integer partition sequence of Hardy--Ramanujan fame:
  $$\displaystyle p(n)= [t^n] \prod_{n\geq 1} \frac{1}{1-t^n} \sim \frac{1}{4n\sqrt3} \exp\left(\pi \sqrt \frac{2n}{3}\,\right)\,.$$ 
    Therefore, this nice closed-form formula of Grossman has many summands if $n$ is large (computing it will have an exponential cost); it is thus useful  to have an algorithmic alternative to it. Bizley's formula~\eqref{Gf} 
allows to compute  $f(an,bn)$ in quasi-linear time by a power series manipulation. This is also the advantage of 
other expressions like the ones given by~\cite{BaFl02} using the kernel method, on which we will come back in the next sections.

Formula~\eqref{Gf} for $n=1$ gives $f(a,b)=\frac{1}{a + b  }\binom{ a + b }{a }$, also known as the rational Catalan numbers $\operatorname{Cat}(a,b)$.
In the last years many properties of the Dyck paths and their ``Catalan combinatorics''  (i.e.,~the enumeration of the numerous combinatorial and algebraic structures related to them) 
were extended to Dyck paths below a line of rational slope. This new area of research is sometimes called ``rational Catalan combinatorics''~\cite{ratcat}.
We expect that the recent developments of ``rational Catalan combinatorics'' have a generalization to $n>1$, but with less simple formulae,
as suggested by Table~\ref{Table}.

\begin{table}[hb!]
\begin{center}\scalebox{.92}{
\renewcommand{\arraystretch}{2}
\begin{tabular}{ |c|c|}
 \hline
                 & \# Dyck walks  from $(0,0)$ to $(an,bn)$ staying weakly below $y=\frac{a}{b} x$\\
                 \hline
$n = 1$ &   $\displaystyle c_1$\\
$n = 2$ &   $\displaystyle c_2+\frac{c_1^2}{2}$\\
$n = 3$ &   $\displaystyle c_3+c_1 c_2 + \frac{c_1^3}{3!}$\\
$n = 4$ &   $\displaystyle c_4+\frac{c_2^2}{2} +c_1 c_3 + \frac{c_1^2 c_2}{2} + \frac{c_1^4}{4!}$\\
$n = 5$ &   $\displaystyle c_5+ c_2c_3+c_1c_4+\frac{c_1 c_2^2}{2} + \frac{c_1^2 c_3}{2} + \frac{c_1^3}{3!} c_2 + \frac{c_1^5}{5!} $\\
$n = 6$ &  $\displaystyle c_6+c_5 c_1 + c_4 c_2 + \frac{c_1^2 c_4}{2}  + \frac{c_3^2}{2} + \frac{c_2^3}{ 3!}  + \frac{c_2 c_1^4}{4!} +  \frac{{c_1}^3 c_3}{3!}+ \frac{{c_1}^2 {c_2}^2}{4} + c_1c_2c_3 + \frac{c_1^6}{6!}$\\
\vdots & \vdots\\
$n$ &  $\displaystyle \sum_{\left\{\substack{\text{integer partitions of } n:\\
\sum_{j=1}^{k}  j \, e_j = n
}\right\}}   \,  \prod_{j=1}^{k}  \frac{(c_{j})^{e_j}}{e_j!}$\\
\hline
\end{tabular}
}
\end{center}
\caption{The number  $f(an,bn)$  of Dyck walks  from $(0,0)$ to $(an,bn)$ staying weakly below $y=\frac{a}{b} x$.
To shorten our expressions, we use the shorthand $c_j := \frac{1}{a j+ b j }\binom{ a j+ b j}{a j}$.}\label{Table}
In the rest of the article, we will see further nice formulae for Dyck paths below a rational slope.
\end{table}

\pagebreak
\section{Knuth's AofA problem \#4}
\label{sec:KnuthProblem}

During the conference ``Analysis of Algorithms'' (AofA'2014) in Paris in June 2014,
Knuth gave the first invited talk, dedicated to the memory of Philippe Flajolet (1948-2011).
The title of his lecture was ``Problems that Philippe would have loved'' and he was pinpointing/developing five nice open problems with a good flavor of ``analytic combinatorics''
(his slides are available online\footnote{\texttt{\url{http://www-cs-faculty.stanford.edu/~uno/flaj2014.pdf}}}).
The fourth problem was on ``Lattice paths of slope 2/5'', in which Knuth investigated Dyck paths under a line of slope 2/5, following the work of~\cite{Nakamigawa12}. 
This is best summarized by the two following original slides of Knuth:

\begin{center}
\begin{figure}[th]
\begin{minipage}{0.5\textwidth}
	\begin{center}
		\fbox{\includegraphics[width=0.949\textwidth]{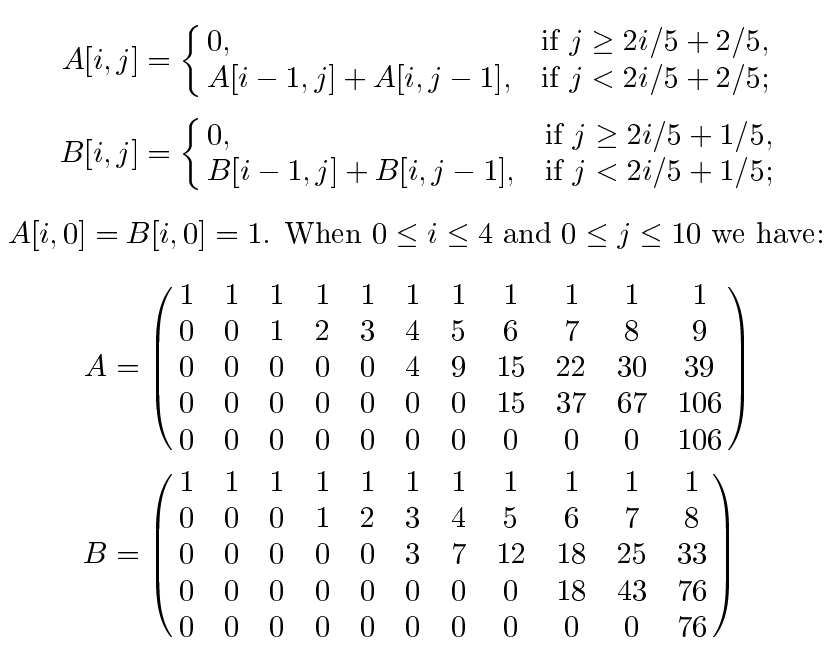}}
		\end{center}
\end{minipage}
\begin{minipage}{0.5\textwidth}
	\begin{center}
		\fbox{\includegraphics[width=0.99\textwidth]{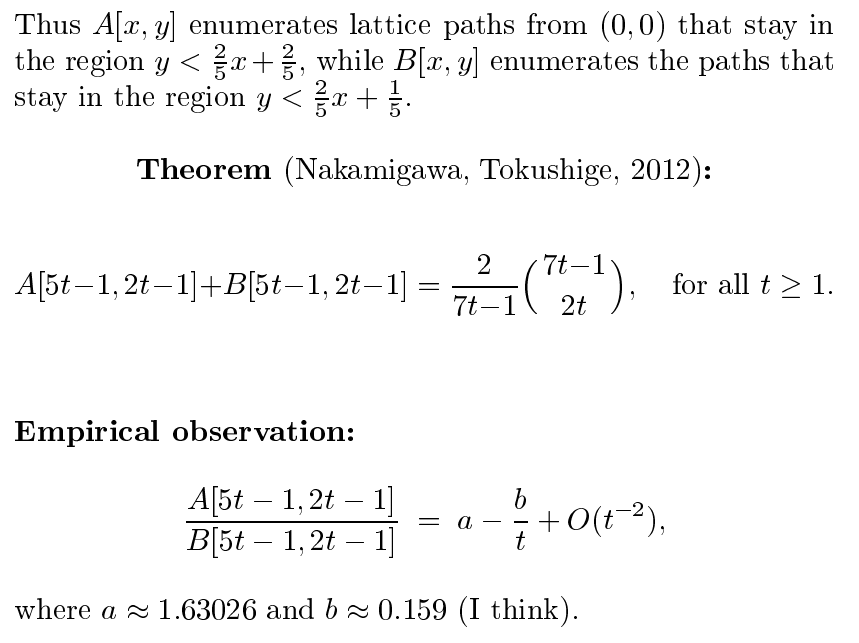}}
		\end{center}
\end{minipage}
\end{figure}
\end{center}

In the next sections we prove that Knuth was indeed right!  In order not to conflict with our notation, 
let us rename Knuth's constants $a$ and $b$ into $\kappa_1$ and $\kappa_2$.

\section{A bijection for lattice paths below a rational slope}
\label{sec:bijection}

Consider paths in the $\N^2$ lattice\footnote{We live in a world where $0\in\N$.}, starting in the origin, and whose allowed steps are of the type either East or North (i.e.,~steps $(1,0)$ and $(0,1)$, respectively). Let $\alpha, \beta$ be positive rational numbers.
We restrict the walks to stay strictly below the barrier $L: y = \alpha x + \beta$. Hence, the allowed domain of our walks forms an obtuse cone with the $x$-axis, the $y$-axis and the barrier~$L$ as boundaries. The problem of counting walks in such a domain is equivalent to counting directed walks in the Banderier--Flajolet model \cite{BaFl02}, as seen via the following bijection:

\pagebreak
\begin{prop}[Bijection: Lattice paths below a rational slope are directed lattice paths]
	\label{prop:bijgen}
	Let $\Dc: y < \alpha x + \beta$ be the domain strictly below the barrier $L$. 
From now on, we assume without loss of generality that $\alpha=a/c$ and $\beta=b/c$ where $a, b, c$ are positive integers such that $\gcd(a,b,c)=1$  (thus, it may be the case that $a/c$ or $b/c$ are reducible fractions).
There exists a bijection between ``walks starting from the origin with North and East steps''
and ``directed walks starting from $(0,b)$ with the step set $\{(1,a), (1,-c)\}$''. What is more, the restriction of staying below the barrier $L$ is mapped to the restriction of staying above the $x$-axis. 	
\end{prop}	
\begin{proof}
The following affine transformation gives the  bijection (see Figure~\ref{fig:bij}):\vspace{-3mm}
	\begin{align*}
		 \begin{pmatrix} x \\y  \end{pmatrix} \mapsto
			\begin{pmatrix}
				x + y \\
				a x - c y + b
			\end{pmatrix}.
	\end{align*}
Indeed, the determinant of the involved linear mapping is $-(c+a) \neq 0$. 
What is more, the constraint of being below the barrier (i.e.,~one has $y<\alpha x+\beta$)
is thus forcing the new abscissa to be positive: $ax-cy+b>0$. The  gcd conditions ensure an optimal choice (i.e.,~the thinnest lattice) for the lattice on which walks will live. 
Note that this affine transformation gives a bijection not only in the case of an initial step set North and East, but for any set of jumps.
\end{proof}

The purpose of this  bijection is to map walks of length $n$ 
to meanders (i.e.,~walks that stay above the $x$-axis) which are constructed by $n$ unit steps into the positive $x$ direction. 

\begin{figure}[!hb]
	\centering
	\subfloat[Rational slope model]{
		\includegraphics[width=0.5\textwidth]{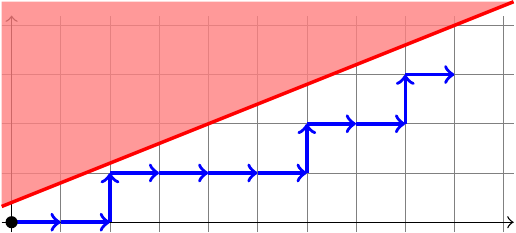}
	}
	\quad
	\subfloat[Banderier--Flajolet model]{
		\includegraphics[width=0.40\textwidth,height=41mm]{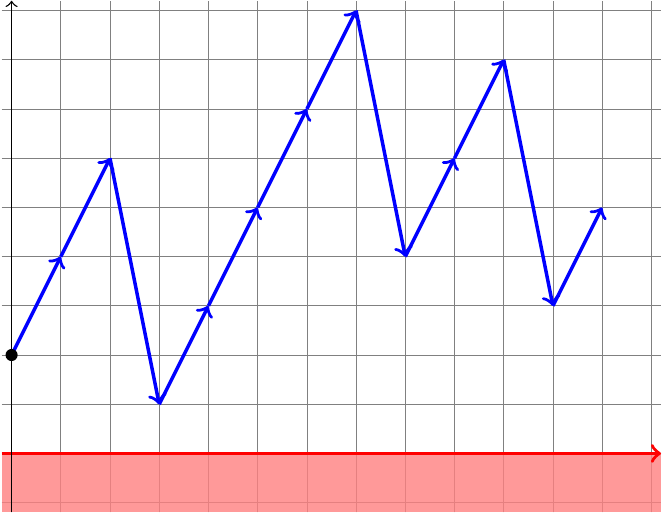}
	}
	\caption{Example showing the bijection from Proposition~\ref{prop:bijgen}: Dyck paths below the line $y=(2/5) x +2/5$ (or touching it)
	are in bijection with walks allowing jumps $+2$ and $-5$, starting at altitude 2, and staying above the line $y=0$ (or touching it).}
	\label{fig:bij}
\end{figure}

Note that if one does not want the walk to touch the line $y=(a/c) x + b/c$,
it corresponds to a model in which one allows to touch, but with a border at $y= (a/c) x + (b-1)/c$.
Time reversal is also giving a bijection between 
\begin{itemize} \item walks starting at altitude $b$ with jumps $+a, -c$ and ending at $0$,
\item  and walks starting at $0$ and ending at altitude $b$ with jumps $-a,+c$.
\end{itemize}

\pagebreak
\section[Lattice paths of slope 2/5]{Functional equation and closed-form expressions for lattice paths of slope 2/5}
\label{sec:funceq}

In this section, we show how to derive closed-forms (i.e.,~explicit expressions) for the generating functions of lattice paths of slope 2/5 (and their coefficients).
First, define the jump polynomial $P(u) := u^{-2} + u^5$. Note that the bijection in Proposition~\ref{prop:bijgen} gives jump sizes $+2$ and~$-5$. 
However, a time reversal gives this equivalent model (jumps $-2$ and $+5$), which has the advantage of leading to more compact formulae (see below). 
Let $f_{n,k}$ be the number of walks of length $n$ which end at altitude $k$. The corresponding bivariate generating function is given by
\begin{align*}
	F(z,u) &= \sum_{n,k \geq 0} f_{n,k} z^n u^k = \sum_{n \geq 0} f_n(u) z^n = \sum_{k \geq 0} F_k(z) u^k,
\end{align*}
where the $f_n(u)$ encode all walks of length $n$, and the $F_k(z)$ are the generating functions of walks ending at altitude $k$. A step-by-step approach yields the following linear recurrence
\begin{align*}
	f_{n+1}(u) &= \{u^{\geq 0}\} \left[ P(u) f_n(u) \right] \qquad \text{ for } n \geq 0,
\end{align*}
with initial value $f_0(u)$ (i.e.,~the polynomial representing the walks of length $0$), and where $\{u^{\geq 0}\}$ is a linear operator extracting all monomials in $u$ 
with non-negative exponents. Summing the terms $z^{n+1} f_{n+1}(u)$ leads to  the functional equation
\begin{align}
	\label{eq:funceq}
	(1 - z P(u)) F(z,u) = f_0(u) - z u^{-2} F_0(z) - z u^{-1} F_1(z).
\end{align}
We apply the \emph{kernel method} in order to transform this equation into a system of linear equations for $F_0$ and $F_1$. The factor $K(z,u):=1-zP(u)$ is called the \emph{kernel} and the kernel equation is given by $K(z,u)=0$. Solving this equation for $u$, we obtain $7$ distinct solutions. These split into two groups, namely, we get $2$ small roots $u_1(z)$ and $u_2(z)$ (the ones going to 0 for $z\sim0$) and $5$ large roots which we call $v_i(z)$ for $i=1,\ldots,5$ (the ones going to infinity for $z\sim0$). It is legitimate to insert the $2$ small branches into~\eqref{eq:funceq} to obtain\footnote{In this article, whenever we thought it could ease the reading, without harming the understanding,
we write $u_1$ for $u_1(z)$, or $F$ for $F(z)$, etc.}
\begin{align*}
	z F_0 + z u_1 F_1 &= u_1^2 f_0(u_1),\\
	z F_0 + z u_2 F_1 &= u_2^2 f_0(u_2).
\end{align*}
This linear system is easily solved by Kramer's formula, which yields
\begin{align*}
	F_0(z) &= - \frac{u_1 u_2 \left(u_1 f_0(u_1) - u_2 f_0(u_2) \right)}{z (u_1 - u_2)}\,, \\
	F_1(z) &= \frac{u_1^2 f_0(u_1) - u_2^2 f_0(u_2) }{z (u_1 - u_2)}\,.
\end{align*}
Now, let the functions $F(z,u)$ and $F_k(z)$ denote functions associated with $f_0(u) = u^3$ (i.e.,~there is one walk of length $0$ at altitude $3$) and let the functions $G(z,u)$ and  $G_k(z)$ denote functions associated with $f_0(u) = u^4$. 
One thus gets the following theorem:
\begin{theo}[Closed-forms for the generating functions]
	\label{theo:closedformgf}
Let us consider walks in $\N^2$ with jumps $-2$ and $+5$.
The number of such walks starting at altitude 3 and ending at altitude 0 is given by $F_0(z)$,
 the number of such walks starting at altitude 4 and ending at altitude~1 is given by $G_1(z)$, and we have the following closed-forms 
in terms of the small roots $u_1(z)$ and $u_2(z)$ of $1-zP(u)=0$ with $P(u)=u^{-2}+u^5$:
\begin{align}
	\label{eq:F0G1}
	F_0(z) &=  - \frac{u_1 u_2 \left(u_1^4 - u_2^4\right)}{z (u_1 - u_2)}, \\
	G_1(z) &= \frac{u_1^6 - u_2^6 }{z (u_1 - u_2)}\,.
\end{align} 
\end{theo}

Thanks to the bijection given in Section~\ref{sec:bijection} between walks in the rational slope model and directed lattice paths in the Banderier--Flajolet model
(and by additionally reversing the time\footnote{Reversing the time allows us to express all generating functions in terms of just $2$ roots. If one does not  reverse time, 
everything works well but the expressions contain the $5$ large roots, yielding more complicated closed-forms.}), it is now possible to relate the quantities $A$ and $B$ of Knuth with $F_0$ and $G_1$:
\begin{align}
	\label{eq:ABdef}
	A_n&:= A[5n-1,2n-1] = [z^{7n-2}] G_1(z), \\
	B_n&:= B[5n-1,2n-1] = [z^{7n-2}] F_0(z).
\end{align}
Indeed, from the bijection of Proposition~\ref{prop:bijgen}, the walks strictly below $y =\frac{a}{c} x + \frac{b}{c}$  
(with $a=2$, $c=5$) 
and ending at $(x,y) = (5n-1,2n-1)$
are mapped (in the Banderier--Flajolet model, not allowing to touch $y=0$) to walks 
starting at $(0,b)$ and ending at $(x+y,ax-cy+b)=(7n-2,3+b)$.
Reversing the time and allowing to touch $y=0$ (thus $b$ becomes $b-1$),
gives that $A_n$ counts walks starting at 4, ending at 1 (yeah, this is counted by $G_1$!) and that $B_n$ counts walks starting at~3, ending at 0 (yeah, this is counted by $F_0$!).
While there is no nice formula for $A_n$ or $B_n$ (see, however,~\cite{BanderierDrmota} and page~\pageref{eq:nestsumAn} for a formula involving nested sums of binomials),
it is striking that there is a simple and nice formula for $A_n+B_n$:
\begin{theo}[Closed-form for the sum of coefficients]
\label{theo:closedformcoeff}
The sum of the number of Dyck paths (in our rational slope model) touching or staying below $y=(2/5)x+1/5$ and $y=(2/5)x$  
simplifies to the following expression:
\begin{align}
	\label{eq:AplusBex}
	A_n + B_n &= \frac{2}{7n-1} \binom{7n-1}{2n}.
\end{align}
\end{theo}
\begin{proof}
A first proof of this was given by~\cite{Nakamigawa12} using a variant of the cycle lemma. (We comment more on this in Section~\ref{sec:nakatoku}.)
We give here another proof, indeed, our Theorem~\ref{theo:closedformgf} (Closed-form for the generating functions) implies that
\begin{align}
	\label{eq:AplusB}
	A_n + B_n &= [z^{7n-1}] \left( u_1^5 + u_2^5 \right)\,.
\end{align}
This suggests using holonomy theory to prove the theorem. First, a resultant equation gives the algebraic equation for $U:=u_1^5$ (namely, $z^7+(U-1)^5 U^2=0$)
and then the  Abel--Tannery--Cockle--Harley--Comtet theorem
(see the comment after Proposition 4 in~\cite{BanderierDrmota}) transforms it into a differential equation for 
the series $u_1^5(z^2)$. It is also the differential equation  (up to distinct initial conditions) for $u_2^5(z^2)$  (as $u_2$ is defined by the same equation as $u_1$), and thus of $u_1^5(z^2)+u_2^5(z^2)$.
Therefore, it directly gives the differential equation for the series $C(z)=\sum_n (A_n+B_n) z^n$,
and it corresponds to the following recurrence for its coefficients:
$${C_{n+1}=\frac{7}{10}\frac{(7n+5)(7n+4)(7n+3)(7n+2)(7n+1)(7n-1)}{(5n+4)(5n+3)(5n+2)(5n+1)(2n+1)(n+1)}C_n\,,}$$
which is exactly  the hypergeometric recurrence for $\frac{2}{7n-1} \binom{7n-1}{2n}$ (with the same initial condition).
This computation takes 1 second on an average computer, 
while, if not done in this way (e.g.,  if instead of the resultant shortcut above, one uses several {\texttt{gfun[diffeq*diffeq]}} or variants of it in Maple, see~\cite{SalvyZimmermann94} for a presentation of the corresponding package),
the computations for such a simple binomial formula 
surprisingly take hours.
\end{proof}

Some additional investigations conducted by Manuel Kauers (private communication) show that this is the only linear combination of $A_n$ and $B_n$ 
which leads to a hypergeometric solution
(to prove this, you can compute a recurrence for a formal linear combination  $r A_n+ s B_n$, 
and then check which conditions it implies on $r$ and $s$ if one wishes the associated recurrence to be of order 1, i.e.,~hypergeometric).
It thus appears that $r A_n+ s B_n$ is generically of order 5, with the exception 
of a sporadic $4A_n-B_n$ which is of order 4,
and the miraculous $A_n+B_n$ which is of order 1 (hypergeometric).

However, there are many other hypergeometric expressions floating around: 
expressions of the type of the right-hand side of~\eqref{eq:AplusB} have nice hypergeometric closed-forms.
This can also be explained in a combinatorial way, indeed we observe that setting $k=-5$ in Formula~(10) from~\cite{BaFl02},
 leads to $5 W_{-5}(z) = \Theta(A(z)+B(z))$ (where $\Theta$ is the pointing operator). 
The ``Knuth pointed walks'' are thus in 1-to-5 correspondence with unconstrained walks 
(see our Table~\ref{fig-4types}, top left) ending at altitude -5.
 
We want to end this chapter with exemplifying the miracles involved in the simplifications of~\eqref{eq:AplusBex}.
Using the Flajolet--Soria formula~\cite{BanderierDrmota} for the coefficients of an algebraic function, we can extract the coefficient of $z^{7n-2}$  of $G_1(z)$ and $F_0(z)$ in terms of nested sums. 
According to~\eqref{eq:ABdef}, this corresponds to $A_n$ and $B_n$, 
which are thus given by formulae involving respectively $45$ and $34$ nested sums (see Figure~\ref{eq:nestsumAn}).

Then, in the next section, we perform some analytic investigations in order to prove what Knuth conjectured:
\begin{align*}
	\frac{A_n}{B_n} & =   \kappa_1 - \frac{\kappa_2}{n} + \LandauO(n^{-2}),
\end{align*}
with $\kappa_1\approx 1.63026$ and $\kappa_2\approx 0.159$.

\pagebreak

\begin{figure}[!ht]
\fbox{\begin{minipage}{\textwidth}
\begin{align*} 
	A_n &= \sum_{m = 0}^{7n-2} m! 
	\sum_{\substack{m_1+\cdots+m_{44} = m+1 \\
	                b_1m_1 + \cdots + b_{44}m_{44} = 7n-2\\
	                c_1m_2+\cdots+c_{44}m_{44} =m}}	
	\Big( 
20^{m_{1}}
3^{m_{2}}
(-190)^{m_{3}}
(-39)^{m_{4}}
1140^{m_{5}}
239^{m_{6}}
4^{m_{7}}
(-4845)^{m_{8}}
	\\&
(-915)^{m_{9}}
(-25)^{m_{10}}
15504^{m_{11}}
2443^{m_{12}}
68^{m_{13}}
1^{m_{14}}
(-38760)^{m_{15}}
(-4806)^{m_{16}}
(-105)^{m_{17}}
	\\&
77520^{m_{18}}
7173^{m_{19}}
100^{m_{20}}
(-125970)^{m_{21}}
(-8238)^{m_{22}}
(-59)^{m_{23}}
167960^{m_{24}}
7305^{m_{25}}
20^{m_{26}}
	\\&
(-184756)^{m_{27}}
(-4971)^{m_{28}}
(-3)^{m_{29}}
167960^{m_{30}}
2553^{m_{31}}
(-125970)^{m_{32}}
(-959)^{m_{33}}
77520^{m_{34}}
	\\&
249^{m_{35}}
(-38760)^{m_{36}}
(-40)^{m_{37}}
15504^{m_{38}}
3^{m_{39}}
(-4845)^{m_{40}}
1140^{m_{41}}
(-190)^{m_{42}}
	\\&
20^{m_{43}}
(-1)^{m_{44}}
		\Pi_{k=1}^{44} \frac{1}{m_{i}!}
	\Big),
\end{align*}

\noindent where $(b_n)_{n=1}^{44}=$ {\small (2,5,4,7,6,9,12,8,11,14,10,13,16,19,12,15,18,14,17,20,16,19,22,18,21,24,20,23,26, 22,25,24,27,26,29,28,31,30,33,32,34,36,38,40)} and 
$(c_n)_{n=1}^{44}=$ {\small 
(2,0,3,1,4,2,0,5,3,1,6,4,2,0,7,5,3,8,6, 4,9,7,5,10,8,6,11,9,7,12,10,13,11,14,12,15,13,16,14,17,18,19,20,21)}.

\begin{align*}
	B_n &= \sum_{m = 0}^{7n-2} m! 
	\sum_{\substack{m_1+\cdots+m_{33} = m+1 \\
	                b_1m_1 + \cdots + b_{33}m_{33} = 7n-2\\
	                c_1m_2+\cdots+c_{33}m_{33} =m}}	
	\Big( 
	20^{m_{1}}
2^{m_{2}}
(-182)^{m_{3}}
(-18)^{m_{4}}
1006^{m_{5}}
73^{m_{6}}
(-1)^{m_{7}}
(-3793)^{m_{8}}
	\\&
(-176)^{m_{9}}
10349^{m_{10}}
279^{m_{11}}
(-21084)^{m_{12}}
(-294)^{m_{13}}
32521^{m_{14}}
190^{m_{15}}
1^{m_{16}}
(-37980)^{m_{17}}
	\\&
(-57)^{m_{18}}
(-10)^{m_{19}}
33128^{m_{20}}
45^{m_{21}}
(-20928)^{m_{22}}
(-120)^{m_{23}}
9039^{m_{24}}
210^{m_{25}}
(-2384)^{m_{26}}
	\\&
(-252)^{m_{27}}
289^{m_{28}}
210^{m_{29}}
(-120)^{m_{30}}
45^{m_{31}}
(-10)^{m_{32}}
1^{m_{33}}
		\Pi_{k=1}^{33} \frac{1}{m_{i}!}
	\Big),
\end{align*}
\noindent where
 $(b_n)_{n=1}^{33}=$ {\small (2,5,4,7,6,9,12,8,11,10,13,12,15,14,17,13,16,19,15,18,17,20,19,22,21,24,23,26,25, 27,29,31,33)} and $(c_n)_{n=1}^{33}=${\small 
(2,0,3,1,4,2,0,5,3,6,4,7,5,8,6,11,9,7,12,10,13,11,14,12,15,13,16,14,17,18, 19,20,21)}.

$$A_n + B_n = \frac{2}{7n-1} \binom{7n-1}{2n}.$$
\end{minipage}}
\caption[Caption for LOF]{\label{eq:nestsumAn} \textbf{The ``ugly + ugly = nice'' formula.}
$A_n$ is counting Dyck paths touching or staying below the line $y=(2/5) x+ 1/5 $, 
and $B_n$ is counting Dyck paths touching or staying below the line $y=(2/5) x$.
They are given by complicated ``ugly'' nested sums\footnotemark,
so the miracle is that the sum $A_n + B_n$ is nice. We give several explanations of this fact in this article. 
}
\end{figure}
\footnotetext{Via the kernel method, as explained in~\cite{KKKK16},
it is possible to express $A_n$ and $B_n$ with less nested sums than in Figure~\ref{eq:nestsumAn} 
but the corresponding formulae are however still of the ``ugly'' type!}

\pagebreak
\section{Asymptotics}
\label{sec:asymptotics}

As usual, we need to locate the dominant singularities, and to understand the local behaviour there. 
The fact that there are several dominant singularities makes the game harder here,
and this case was only sketched in~\cite{BaFl02}.
Similarly to what happens in the rational world (Perron--Frobenius theory), or in the algebraic world (see \cite{BanderierDrmota}), 
a periodic behaviour of the generating function leads to some more complicated proofs,
because additional details have to be taken into account. With respect to walks,
it is e.g.\ crucial to understand how singularities spread amongst the roots of the kernel.
To this aim,
some quantities will play a key role: the structural constant $\tau$ is defined as the unique positive root of $P'(\tau)$,
where $$P(u)=u^{-2}+u^5$$ is encoding the jumps, and the structural radius $\rho$ is given as $\rho = 1/P(\tau)$. For our problem, one thus has the explicit values:
\begin{align*}
	\tau    &= \sqrt[7]{\frac{2}{5}}, & 
	P(\tau) &= \frac{7}{10} \sqrt[7]{2^5 5^2}, &
	\rho    &= \frac{\sqrt[7]{2^2 5^5}}{7}.
\end{align*}

\begin{figure}[h!]
	\centering
\includegraphics[width=0.35\textwidth,height=51mm]{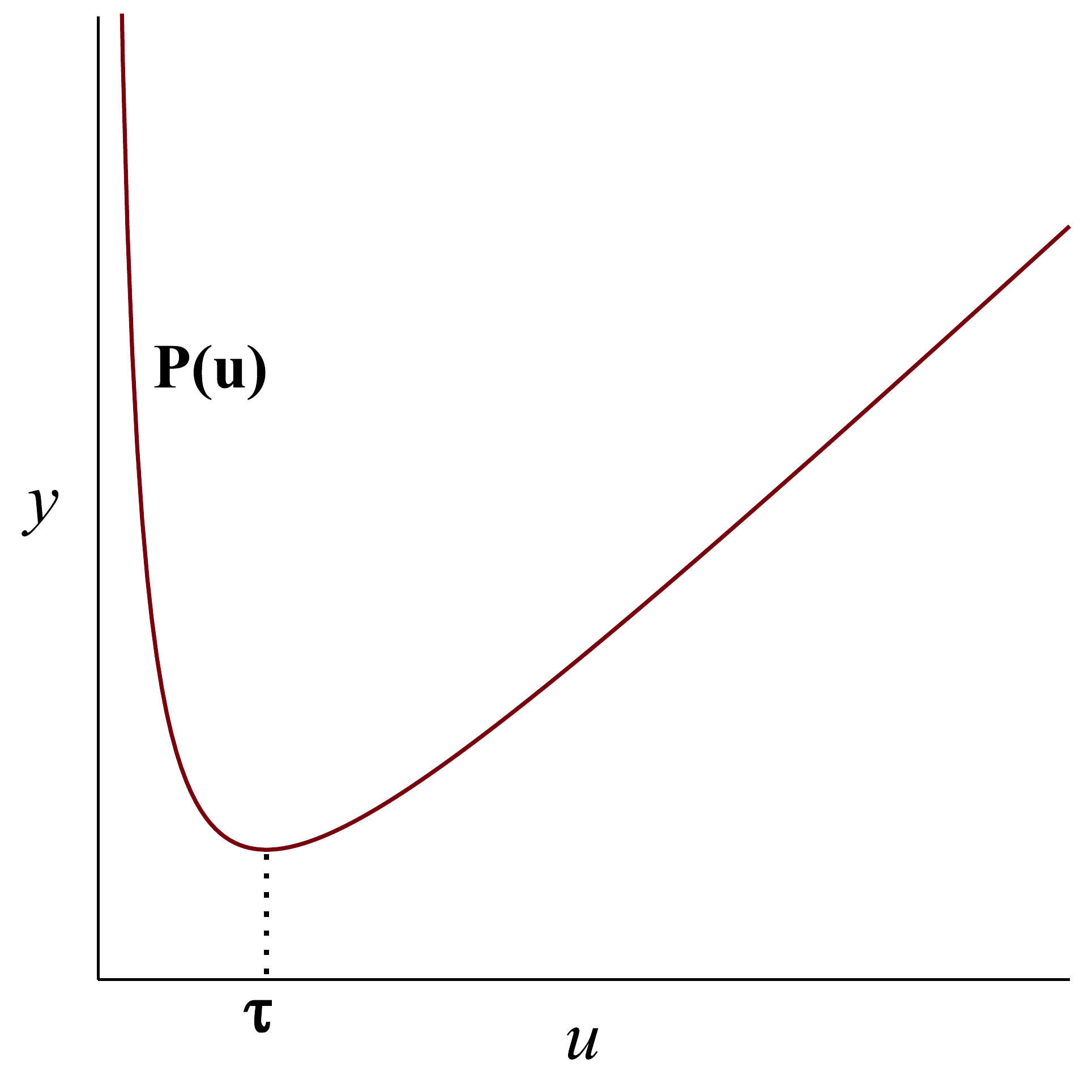} \qquad \includegraphics[width=0.35\textwidth,height=51mm]{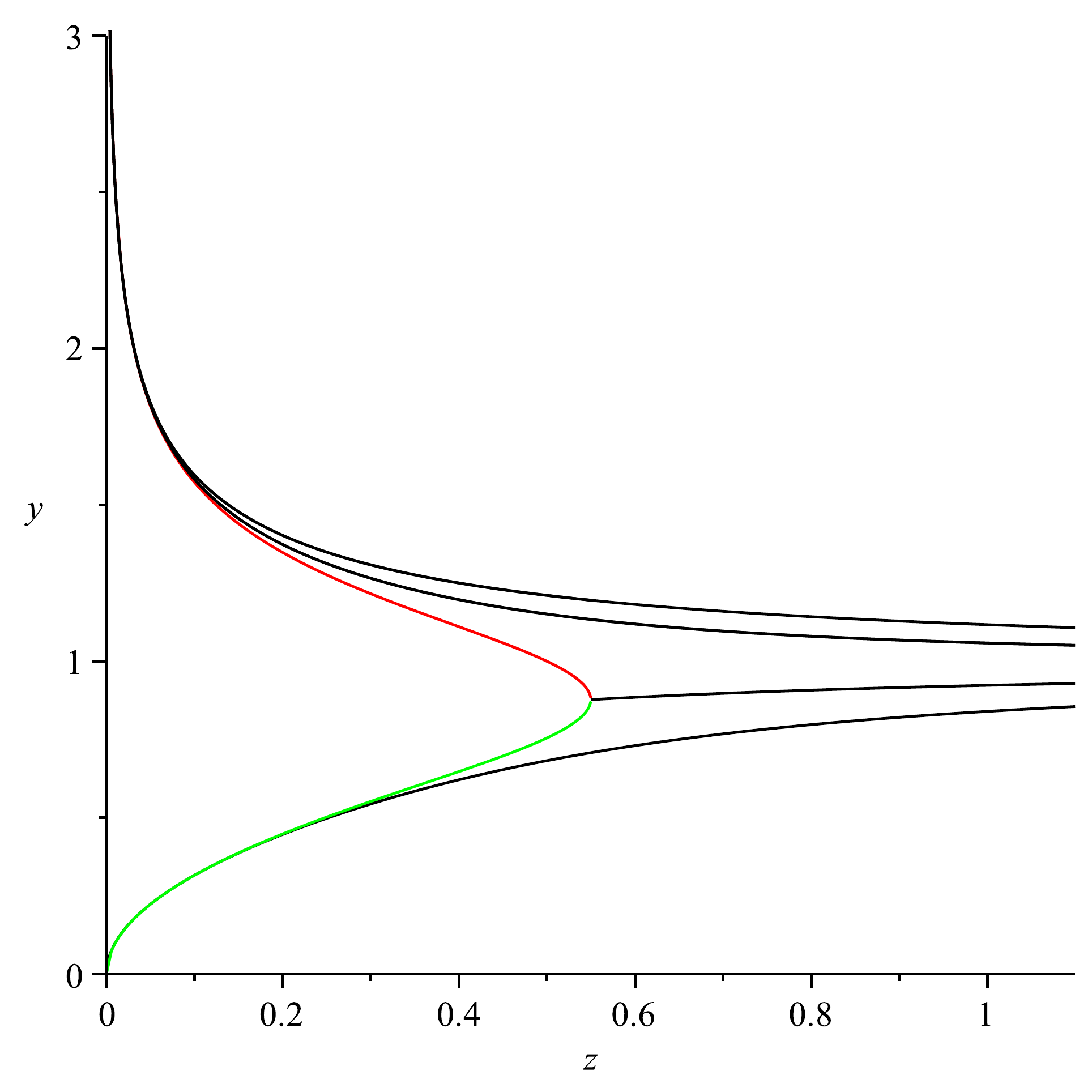}
\caption{$P(u)$ is the polynomial encoding the jumps, its saddle point $\tau$ gives the singularity $\rho=1/P(\tau)$ 
where the small root $u_1$ (in green) meets the large root $v_1$ (in red), with a square root behaviour. (In black, we also plotted $|u_2|, 
|v_2|=|v_3|$, and $|v_4|=|v_5|$.)
This is the key for all asymptotics of such lattice paths.
 }
\end{figure}

From~\cite{BaFl02}, we know that the small branches $u_1(z)$ and $u_2(z)$ are possibly singular only at the roots of $P'(u)$. 
Note that the jump polynomial has \emph{periodic support} with period $p=7$ as $P(u) = u^{-2} H(u^7)$ with $H(u) = 1 + u$. Due to that, there are $7$ possible singularities of the small branches
\begin{align*}
	\zeta_k &= \rho \omega^k, \qquad \text{ with } \omega = e^{2 \pi i / 7}.
\end{align*}
\begin{definition}
	We call a function $F(z)$ \emph{$p$-periodic} if there exists a function $H(z)$ such that $F(z)=H(z^p)$.
\end{definition}

Additionally, we have the following local behaviours:
\begin{lemma}[Local behaviour due to rotation law]
	\label{lem:u1u2}
	The limits of the small branches when $z \to \zeta_k$ exist and are equal to
	\begin{align*}
		u_1(z) \underset{z\,\sim \,\zeta_k}{=} &
				\begin{cases}
				\tau \omega^{-3k} + C_k \sqrt{1 - z / \zeta_k} +\LandauO((1 - z / \zeta_k)^{3/2}), &
				\text{for } k = 2,5,7,\\     
					\tau_2 \omega^{-3k}   + D_k (1-z / \zeta_k) + \LandauO((1 - z / \zeta_k)^2), &
					\text{for } k= 1,3,4,6,
				\end{cases}\\
		u_2(z) \underset{z\,\sim \,\zeta_k}{=}  &
				\begin{cases}
					\tau_2 \omega^{-3k}  + D_k (1-z / \zeta_k) + \LandauO((1 - z / \zeta_k)^2), &
					\text{for } k = 2,5,7,\\     
					\tau \omega^{-3k} +  C_k \sqrt{1 - z / \zeta_k} +\LandauO((1 - z / \zeta_k)^{3/2}), &
					\text{for } k= 1,3,4,6,
				\end{cases}
	\end{align*}
	where $\tau_2 = u_2(\rho)\approx -.707723271$ is the unique real root of
	$500t^{35}+3900t^{28}+13540t^{21}+27708t^{14}+37500t^7+3125$,
		$C_k	   = -\frac{\tau}{\sqrt{5}} \omega^{-3k}$, and $D_k = \tau_2\frac{\tau_2^7+1}{5\tau_2^7-2} \omega^{-3k}$.
\end{lemma}
\begin{proof}
We will show the following \emph{rotation law} for the small branches
	(for all $z \in \C$, with $|z| \leq \rho$ and $0<\arg(z)<\pi-2\pi/7$):
	\begin{align*}
	           u_1(\omega z) &=  \omega^{-3} u_2(z),\\ 
		u_2(\omega z) &= \omega^{-3} u_1(z). 
	\end{align*}  
Indeed, let us consider the function $U(z):=\omega^3 u_i(w z)$ (with $i=1$ or $i=2$, as you prefer!) and a mysterious quantity $X$,
defined by $X(z):= U^2-z\phi(U)$ (where $\phi(u):=u^2 P(u)$).
So we have $X(z)= (\omega^3 u_i(\omega z))^2-z \phi(\omega^3 u_i(\omega z))=  \omega^6 u_i(\omega z)^2 -z \phi(u_i(\omega z))$
(because $\phi$ is $7$-periodic)
and thus $\omega X(z/\omega)=  \omega(\omega^6 u_i(z)^2 -z/\omega \phi(u_i(z))) =  u_i(z)^2 -z \phi(u_i(z))$,
which is $0$ because we recognize here the kernel equation.
This implies that $X=U^2-z\phi(U)=0$ and thus $U$ is a root of the kernel.
Which one? It is one of the small roots, because it is converging to $0$ at $0$.
What is more, this root $U$ is not $u_i$, because it has a different Puiseux expansion 
(and Puiseux expansions are  unique). So, by the analytic continuation principle (therefore, here, as far as we avoid the cut line 
$\arg(z)=-\pi$), we just proved that $\omega^3 u_1(\omega z)= u_2(z)$ and $\omega^3 u_2(\omega z)= u_1(z)$
(and this also proves a similar rotation law for large branches, but we do not need it).

Accordingly, at every $\zeta_k$, amongst the two small branches, only one branch becomes singular: this is $u_1$ for $k=2,5,7$ and $u_2$ for $k=1,3,4,6$. 
This is illustrated in Figure~\ref{fig:singui}.
\begin{center}
	\begin{figure}[ht!]		
	\centering
	\def\n{7}
	\scalebox{1.19}{
	\begin{tikzpicture}[dot/.style={draw,fill,circle,inner sep=1pt}]
	  \pgfmathsetmacro{\ns}{\n-1}
	  \draw[->] (-1.7,0) -- (1.7,0) node[above] {$\Re$};
	  \draw[->] (0,-1.7) -- (0,1.7) node[left] {$\Im$};
	  \draw[help lines] (0,0) circle (1);

	  \foreach \i in {0,...,\ns} {		
		 \ifthenelse{\i=0}{
			\pgfmathsetmacro{\curdot}{\i*360/\n-45}}{
			\pgfmathsetmacro{\curdot}{\i*360/\n}};
		 \node[dot,label={\curdot:$\zeta_{\i}$}] (w\i) at (\i*360/\n:1) {};
		 \draw[-] (0,0) -- (w\i);
	  }
	  \draw[->] (0:.3) arc (0:360/\n:.3);
	  \node at (360/\n/2:.5) {\scriptsize$\frac{2\pi}{\n}$};
	\end{tikzpicture}
	}
	\qquad \qquad
	\def\n{7}
	\scalebox{1.19}{
	\begin{tikzpicture}[dot/.style={draw,fill,circle,inner sep=1pt}]
		\pgfmathsetmacro{\ns}{\n-1}
		\draw[->] (-1.7,0) -- (1.7,0) node[above] {$\Re$};
		\draw[->] (0,-1.7) -- (0,1.7) node[left] {$\Im$};
		\draw[help lines] (0,0) circle (1);

		\foreach \i in {0,...,\ns} {		
		 \ifthenelse{\i=0}{
			\pgfmathsetmacro{\curdot}{\i*360/\n-45}}{
			\pgfmathsetmacro{\curdot}{\i*360/\n}};
			
		 \ifthenelse{\i=0 \OR \i=2 \OR \i=5}{
			\node[dot,label={\curdot:\textcolor{blue}{$u_{1}$}}] (w\i) at (\i*360/\n:1) {};}{
			\node[dot,label={\curdot:\textcolor{red}{$u_{2}$}}] (w\i) at (\i*360/\n:1) {};};					
		 \draw[-] (0,0) -- (w\i);
		}
		\draw[->] (0:.3) arc (0:360/\n:.3);
		\node at (360/\n/2:.5) {\scriptsize$\frac{2\pi}{\n}$};
	\end{tikzpicture}
	}
	\caption{
		The locations of the $7$ possible singularities of the small branches (left); the small branch which is singular at that location (right).
	}
	\label{fig:singui}
	\end{figure}
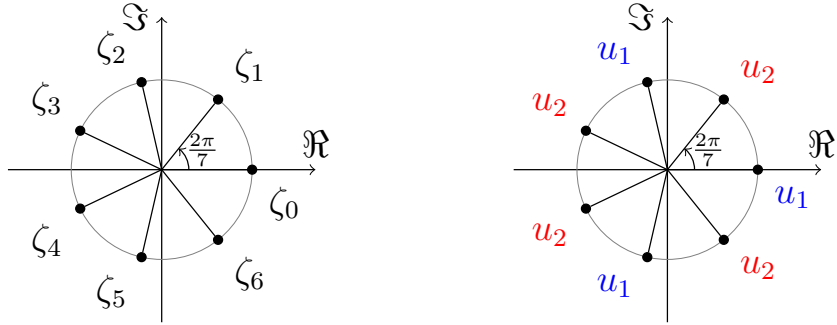
\end{center}
Hence, we directly see how the asymptotic expansion at the dominant singularities are correlated with the one of $u_1$ at $z=\rho=\zeta_7$,
which we derive following the approach of \cite{BaFl02}; this gives for $z\sim \rho$:
\begin{align*}
	u_{1}(z) &= \tau  + \coefb \sqrt{1 - z / \rho}  + \coefd (1-z/\rho)^{3/2}+\ldots,
\end{align*}	
where $\coefb = -\sqrt{2 \frac{P(\tau)}{P''(\tau)}}$. 
Note that in our case $P^{(3)}(\tau)=0$ (this funny cancellation holds for any $P(u)=p_{5} u^5+p_0+p_{-2} u^{-2}$ ), so even the formula for $\coefd$ is quite simple: $\coefd=-\frac{1}{2} \coefb$.  

In the lemma, the formula for $\tau_2=u_2(\rho)$ is obtained by a resultant computation.\end{proof}

For the local analysis of Knuth's generating functions $F_0(z)$ and $G_1(z)$ with periodic support, we introduce a shorthand notation:

\begin{definition}[Local asymptotics extractor 	$\exzkn$]
Let $F(z)$ be an algebraic function with~$p$ dominant singularities $\zeta_k$  (for $k=1,\ldots, p$). Accordingly, for each $\zeta_k$, $F(z)$ can be expressed as a Puiseux series, i.e.,~there exist $r \in \Q$ and coefficients $c_n$ (both depending on $k$) such that
$$F(z) = \sum_{j \geq 0} c_j (1 -  z/\zeta_k)^{r j}, \qquad \text{for } z \sim \zeta_k.$$
Then we define the local asymptotic extractor $\exzkn$ as
\begin{align*}
	\exzkn F(z) &:= \sum_{j \geq 0} c_j [z^n] (1 -  z/\zeta_k)^{r j}. 
\end{align*}
This notation can be considered as ``extracting the $z^n$-coefficient in the Puiseux expansion\footnote{In fact this notation holds for singular expansions of alg-log functions~\cite{flaj09}, exp-log functions, and more generally for expansions in Hardy fields~\cite{hard17} which are amenable to singularity analysis or saddle point methods.} of $F(z)$ at $z=\zeta_k$'', and singularity analysis  allows to write $[z^n] F(z) = \sum_{k} \exzkn F(z)+o(C^{-n})$, for some constant $C > |\zeta_k|$.
\end{definition}

\begin{examplenoend}
	A sloppy but easy to remember formulation would be to say
	\begin{align*}
		\exzkn F(z) &:= [z^n] \text{(singular expansion of $F(z)$ at $z=\zeta_k$)}.
	\end{align*}
	This is well illustrated by the generating function $D(z)$ of Dyck paths defined by the functional equation $D(z) = 1+ z^2 D(z)^2$. In this case, we have $D(z) = \frac{1 - \sqrt{1- 4z^2}}{2z^2}$ with $p=2$ and $\zeta_1 = 1/2$ and $\zeta_2 = -1/2$. Therefore we get for any $\varepsilon>0$
	\begin{align*}
		[z^n] D(z) &= \exzk{n}{1/2} \,D(z) + \exzk{n}{-1/2} \,D(z) + o\left(\left(2-\varepsilon \right)^{n}\right)\\
		            &= [z^n] (-2\sqrt{2})\sqrt{1 - 2z} + [z^n] (-2\sqrt{2}) \sqrt{1 + 2z} + O\left(\frac{2^n}{n^{5/2}}\right) + o\left(\left(2-\varepsilon \right)^{n}\right) \,. \\
		            \noalign{\vspace{-\baselineskip} \hfill $_\blacksquare$ \vspace{-\baselineskip}}
	\end{align*}
\end{examplenoend}

\begin{prop} [Periodic rule of thumb]
	\label{prop:periodicasympt}
Let $\rho$ be the positive real dominant singularity in the previous definition.
If additionally the generating function $F(z)$ satisfies a rotation law $F(\omega z)=\omega^mF(z)$ (where $\omega =\exp(2i\pi/p)$, $p$ maximal), then one has a neat simplification:
	\begin{align*}
 	 [z^n]F(z) &= p [z^n]_{\rho}  F(z) + o(\rho^n), 
 \end{align*}
if $n-m$ is a multiple of $p$. (The other coefficients are equal to $0$.)
\end{prop}

\begin{proof}
As $F(z)$ is a generating function, it has real positive coefficients and therefore, by Pringsheim's theorem \cite[Theorem IV.6]{flaj09}, 
one of the $\zeta_k$'s  has to be real positive, called $\rho$. We relabel the $\zeta_k$'s such that $\zeta_k:=\omega^k \rho$. Then
\begin{align*}
 	 [z^n]F(z)-o(\rho^n) &=\sum_{k=1}^p \exzkn F(z) 
	                      =\sum_{k=1}^p   \exzkn (\omega^m)^k   F(\omega^{-k} z)  
 							    =\sum_{k=1}^p (\omega^m)^k  (\omega^{-k})^n  [z^n]_\rho  F(z)  \\
    							&=\left(\sum_{k=1}^p   (\omega^k)^{m-n}\right)     [z^n]_\rho  F(z)
  								 = p  [z^n]_{\rho}  F(z), 
 \end{align*}
if $n-m$ is a multiple of $p$, and $0$ elsewhere.
\end{proof}

We can apply this proposition to $F_0(z)$ and $G_1(z)$, because the rotation law for the $u_i$'s 
implies: $F_0(\omega z)=\omega^{-2} F_0(z)$ and $G_1(\omega z)=w^{-2} G_1(z)$. 
Thus, we just have to compute the asymptotics coming from the Puiseux expansion of  $F_0(z)$ and~$G_1(z)$ at $z=\rho$,
and multiply it by 7
(recall that it is classical to infer the asymptotics of the coefficients from the Puiseux expansion of the functions via the so-called ``transfer'' Theorem VI.3 from~\cite{flaj09}), this gives:
\begin{theo}[Asymptotics of coefficients, answer to Knuth's problem] The asymptotics for the number of excursions below $y=(2/5)x+2/5$ and $y=(2/5)x+1/5$  are given by:
\begin{align*}
	A_n &=[z^{7n-2}] G_1(z) = 
		 \alpha_1 \frac{\rho^{-7n}}{\sqrt{\pi (7n-2)^3}} +  \frac{3\alpha_2}{2}\frac{\rho^{-7n}}{\sqrt{\pi (7n-2)^5}} + \LandauO(n^{-7/2}), \\
	B_n &=[z^{7n-2}] F_0(z) = 
		 \beta_1 \frac{\rho^{-7n}}{\sqrt{\pi (7n-2)^3}} + \frac{3\beta_2}{2}\frac{\rho^{-7n}}{\sqrt{\pi (7n-2)^5}} + \LandauO(n^{-7/2}),
\end{align*}
with the following constants where we define the shorthand $\mu:=\tau_2/\tau$:
\begin{align*}
	\alpha_1 &= \frac{\mu^4 + 2 \mu^3 + 3 \mu^2 + 4 \mu + 5}{\sqrt{5}}, \qquad
	\beta_1 = \sqrt{5}-\alpha_1, \qquad
	\beta_2 =  -\frac{9}{10}\sqrt{5} - \alpha_2, \\
	\alpha_2 &= -\frac{1}{10} \frac{5 \tau_2^7 (13\mu^4 + 22\mu^3 + 29\mu^2+36\mu+45) + 2(15\mu^4+20\mu^3+13\mu^2-8\mu-45)}{\sqrt{5}(5\tau_2^7-2)}.
\end{align*}
\end{theo} 
This theorem leads to the following asymptotics for $A_n+B_n$ 
(and this is for sure a good sanity test, coherent with a direct application of Stirling's formula to the closed-form formula~\eqref{eq:AplusBex} for $A_n+B_n$):
\begin{align*}
	A_n+B_n &= \sqrt{\frac{5}{7^3\pi}} \frac{\rho^{-7n}}{\sqrt{ n^3}} + \LandauO(n^{-5/2}).
\end{align*}
Finally, we directly get
\begin{align*}
	\frac{A_n}{B_n} &= \frac{\alpha_1 + \frac{3\alpha_2}{2 (7n-2)}}{\beta_1 + \frac{3\beta_2}{2 (7n-2)}} + \LandauO(n^{-2})  
					     = \frac{\alpha_1}{\beta_1} + \frac{3}{14} \left(\frac{\alpha_2 \beta_1 - \alpha_1 \beta_2}{\beta_1^2} \right)\frac{1}{n} + \LandauO(n^{-2}),
\end{align*}
which implies that Knuth's constants are 
\begin{align*}
	\kappa_1 &= \frac{\alpha_1}{\beta_1} = - \frac{5}{\mu^4 + 2 \mu^3 + 3 \mu^2 + 4 \mu} - 1 \\
	         &\approx 1.6302576629903501404248,\\
          	\kappa_2 &= - \frac{3}{14} \left(\frac{\alpha_2 \beta_1 - \alpha_1 \beta_2}{\beta_1^2} \right) 
          	          = \frac{3}{9800} (13-236 \kappa_1 -194 \kappa_1^2 -388 \kappa_1^3 +437 \kappa_1^4) \\
          	         &\approx         0.1586682269720227755147.
\end{align*}

Now a few resultant computations give the algebraic equations satisfied by $\tau_2$, $\kappa_1$, and~$\kappa_2$. 
We will illustrate their derivation with the required Maple commands. In what follows, these are always set in a typewriter font.
First, we compute an annihilating polynomial for $\rho$:
\begin{maplegroup}
\begin{mapleinput}
\mapleinline{active}{1d}{R1:=resultant(numer(1-z*P),numer(diff(P,u)),u);
}{}
\end{mapleinput}
\begin{maplelatex}
\mapleinline{inert}{2d}{R1 := 823543*z^7-12500}{\[\displaystyle \emph{R1}\, := \,823543\,{z}^{7}-12500\]}
\end{maplelatex}
\end{maplegroup}

Then, we construct from it an annihilating polynomial for $u_i(\rho)$. 

\begin{maplegroup}
\begin{mapleinput}
\mapleinline{active}{1d}{R2:=factor(resultant(numer(1-z*P),R1,z));
}{}
\end{mapleinput}
\mapleresult
\begin{maplelatex}
\mapleinline{inert}{2d}{(500*u^35+3900*u^28+13540*u^21+27708*u^14+37500*u^7+3125)*(-2+5*u^7)^2}{\[\displaystyle  \left( 500\,{u}^{35}+3900\,{u}^{28}+13540\,{u}^{21}+27708\,{u}^{14}+37500\,{u}^{7}+3125 \right)  \left( -2+5\,{u}^{7} \right) ^{2}\]}	
\end{maplelatex}
\end{maplegroup}

This polynomial contains $u_1(\rho)=\tau$, and $u_2(\rho)=\tau_2$ as roots. It factorizes into smaller polynomials and these two roots are in separate factors. Thus, we can go on with the right factor which we save in \texttt{Rtau2}. Then, we continue with the annihilating polynomial for $\mu$.

\begin{maplegroup}
\begin{mapleinput}
\mapleinline{active}{1d}{
resultant(x*t-t2,subs(u=t,diff(P,u)),t);
factor(resultant(\%,subs(u=t2,Rtau2),t2));
}{}
\end{mapleinput}
\end{maplegroup}

We identify the algebraic relation for $\mu$ and save it in \texttt{Rmu}. Finally, we compute the minimal polynomial for $\kappa_1$:

\begin{maplegroup}
\begin{mapleinput}
\mapleinline{active}{1d}{Rmu:=2*u\symbol{94}5+4*u\symbol{94}4+6*u\symbol{94}3+8*u\symbol{94}2+10*u+5;
Rk1:=resultant((x+1)*(u\symbol{94}4+2*u\symbol{94}3+3*u\symbol{94}2+4*u)+5,Rmu,u):
factor(Rk1/igcd(coeffs(Rk1)));
}{}
\end{mapleinput}
\mapleresult
\begin{maplelatex}
\mapleinline{inert}{2d}{-23*x^5+41*x^4-10*x^3+6*x^2+x+1}{\[\displaystyle -23\,{x}^{5}+41\,{x}^{4}-10\,{x}^{3}+6\,{x}^{2}+x+1\]}
\end{maplelatex}
\end{maplegroup}

In conclusion, $\kappa_1$ is the unique real root of the polynomial 
$23 x^5 -41 x^4 +10 x^3 -6x^2 - x- 1$, and similar computations show that $(7/3) \kappa_2$ 
is the unique real root of $11571875x^5-5363750x^4+628250x^3-97580x^2+5180x-142$.
The Galois group of each of these polynomials is $S_5$. This implies that there is no closed-form formula for the Knuth constants $\kappa_1$ and $\kappa_2$ in terms of basic operations on integers, and roots of any degree.

\bigskip

In the next section we want to establish a link with the results from Nakamigawa and Tokushige. We will show how Knuth derived his problem and how to establish more such nice identities. 

\pagebreak

\section[Links with Nakamigawa/Tokushige]{Links with the work of Nakamigawa and Tokushige}
\label{sec:nakatoku}

In this section, we show the connection between a result of Nakamigawa and Tokushige~\cite{Nakamigawa12} and Knuth's statement. Furthermore, we derive extensions of this result.

Let $\alpha, \beta$ be positive rational numbers. The Nakamigawa--Tokushige  model consists of a single boundary $L:y = \alpha x + \beta$, and a lattice point\footnote{In the article~\cite{Nakamigawa12}, $Q=(m,n)$; we changed these coordinates
in order to avoid a conflict with our other notations.} $Q=(q_1, q_2) \in \Z^2$ on $L$, i.e.,~$q_2 = \alpha q_1 + \beta$. Furthermore the walks go in the opposite direction, i.e.,~they start in $Q$, use unit steps South and West (i.e.,~$(0,-1)$ and $(-1,0)$, respectively), and end in the origin. Let $V$ be the ``vast'' set of such walks without any restriction. The enumeration of $V$ is a folklore result: $|V| = \tbinom{q_1+q_2}{q_1}$. Let $W \subset V$ be the set of walks which do not cross the line $L$ and touch it only at $Q$. 

\begin{definition} [Nearest distance to the boundary]
	Let $w \in V$ be a walk from a point $Q$ to the point $(0,0)$. We define the \emph{minimum $y$-distance} $\delta(w)$ as follows: if the walk $w$ touches or crosses the boundary $y=\alpha x +\beta $ after the first step, then let $\delta(w) = 0$, otherwise let $\delta(w)$ be the minimum of $\alpha p_1 + \beta- p_2$, where $(p_1,p_2)$ runs over all lattice points on $w$ except $Q$, see Figure~\ref{fig:minydist}.
\end{definition}

\begin{figure}[ht!]
	\centering
		\includegraphics[width=0.29\textwidth]{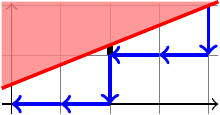}
	\qquad
		\includegraphics[width=0.29\textwidth]{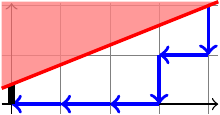}
	\qquad
		\includegraphics[width=0.29\textwidth]{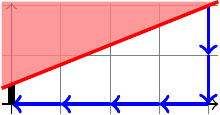}
	\caption{The $3$ walks of length $6$ in the $(2/5)x+2/5$ model with $\delta(w) >0$. The vertical bars mark the minimal $y$-distance $\delta(w)$. The first walk has $\delta(w)=1/5$, whereas the last two have $\delta(w)=2/5$. All of them are members of $W_{1/5}$, but only the two last ones belong to $W_{2/5}$.}
	\label{fig:minydist}
\end{figure}

Hence, we see that $\delta(w) = 0$ if and only if $w \in V \setminus W$, and so $\sum_{w \in V} \delta(w) = \sum_{w \in W} \delta(w)$. Note, if $\alpha$ and $\beta$ are positive integers, then $\sum_{w \in V} \delta(w) = |W|$, because $\delta(w) = 1$ for all $w \in W$. This gives rise to the interpretation as a weighted sum corresponding to the number of walks.

For a real $t \geq 0$, let $W_t := \{ w \in W \, | \, \delta(w) \geq t \}$, i.e.,~the walks staying at least a $y$-distance of $t$ away from the boundary. Due to the definition, $|W_t|$ is a left-continuous step function of~$t$, and we get the representation
\begin{align*}
	\int_{0}^{1} |W_t| \, dt &= \sum_{w \in V} \delta(w).
\end{align*}

It is quite nice that this sum can be further simplified;  this is what the next theorem states:

\begin{theo}[Nakamigawa--Tokushige lattice path integral]
	\label{theo:naka1}
	Let $q_1,q_2$ be positive integers, and let $\alpha,\beta$ be positive reals with $q_2=\alpha q_1 + \beta$. Let $V$ be the set of walks from the origin to the point\footnote{Nota bene: As proven in Lemma~\ref{lem:rationalslope} (Possible starting points on the boundary), if $\alpha$ or $\beta$ are irrational, then there is at most one such point. While if $\alpha$ and $\beta$ are rational (with the right gcd condition), then there are infinitely many such points.}
	$(q_1,q_2)$. Then, we have
	\begin{align}
		\label{eq:naka}
		\int_{0}^{1} |W_t| \, dt &= \sum_{w \in V} \delta(w) = \frac{\beta}{q_1+q_2} \dbinom{q_1+q_2}{q_1}.
	\end{align}
\end{theo}
\begin{proof}
This corresponds to \cite[Theorem 1 and Corollary 1]{Nakamigawa12}, where it is proven using a cycle lemma approach. We give a generalization of this formula in the Section~\ref{sec:duchon} hereafter, based on our kernel method approach,
and Lagrange inversion.
\end{proof}

\textbf{A geometric bijection.}
If $\alpha$ is a rational slope, i.e.,~$\alpha = a/c$ for some $a,c \in \N \setminus \{ 0 \}$, then
\begin{align}
	\label{eq:nakaDiscrete}
	\int_{0}^{1} |W_t| \, dt &= \frac{1}{c} \sum_{t \in T} |W_t|,
\end{align}
where $T = \{\delta(w) \, | \, w \in W \} = \{ 1/c, 2/c, \ldots, (c-1)/c \}$.

This gives rise to the following interpretation:\footnote{In the original work, a slightly different interpretation is given.} If $w \in W$ then the first step is a South step. Then, let $\tilde{w}$ be the walk obtained from $w$ by omitting this step. Therefore, $\tilde{w}$ is a walk with $q_1+q_2-1$ steps, starting from $Q - (0,1) = (q_1,q_2-1)$, and ending in the origin. We see that all these walks which never cross or touch $L$ are in bijection with all walks in $W$. Now, take a walk $w \in W_t$ and its corresponding walk $\tilde{w}$. As $\delta(w) \geq t$, we can translate the barrier $L$ by $t- 1/c$ down and the walk $\tilde{w}$ still does not touch or cross this new barrier $\tilde{L}$. Hence, all walks in $W_t$ are in bijection with walks from $(q_1,q_2-1)$ to the origin which stay strictly below the barrier $\tilde{L}$. 

\begin{example}
This is the bijection that Knuth used in order to state his conjecture. In his case, we have $\alpha=\beta=2/5$ and $q_1=5n-1$, $q_2=2n$ for $n \in \N \setminus \{ 0 \}$. We see that $q_2=\alpha q_1 + \beta$. Hence, $a=2$ and $c=5$ which implies $T = \{1/5, 2/5, 3/5, 4/5 \}$. In this case, the values $3/5$ and $4/5$ are playing no role, as $|W_{3/5}| = |W_{4/5}| = 0$ because $\beta=2/5$ is the maximal value for $\delta(w)$ for all walks to the origin. Therefore, $\int_{0}^{1} |W_t| \, dt$ can be represented by two summands involving $W_{1/5}$ and $W_{2/5}$. 
They correspond to the two models $A$ and $B$ with the barriers $L_{1} : y < (2/5) x + 2/5$ and $L_{2} : y < (2/5) x + 1/5$, respectively where the paths start at $(5n-1,2n-1)$ and move by South and West steps to the origin.
Compare also Figure~\ref{fig:minydist}. Note that in Knuth's case the walks move in the opposite direction, which is obviously equivalent.
\end{example} 

In general, the number of summands $|W_t|$, which corresponds to the number of models in the equivalent formulation, is determined by the size of $T$ minus the maximal $y$-distance at $(0,0)$. Hence, we need to consider $\widetilde{T} = \{t \in T \, | \, t < \beta \} = \{ 1/c,\ldots,k/c \}$. This gives $k$ models with walks from $(q_0,q_1-1)$ to the origin which stay strictly below the boundaries $L_i : y < \alpha x + (\beta - (i-1)/c)$ for $i=1,\ldots,k$. Then, the above reasoning implies that the walks with boundary $L_i$ correspond to the set $W_{i/c}$. Thus, counting the walks in these $k$ models and summing them up, gives the binomial closed-form appearing in the lattice path integral theorem \eqref{eq:naka} divided by $c$, compare with \eqref{eq:nakaDiscrete}.

Up to now in this section, we explained which different counting models are connected with the Nakamigawa--Tokushige lattice path integral formula.
Now, we discuss the possible starting points on the boundary and their interplay with the (ir)rationality of the slope.

\begin{lemma}[Possible starting points on the boundary]
	\label{lem:rationalslope}
	Let $\alpha,\beta$ be positive reals. Then the equation $y = \alpha x + \beta$ possesses in the positive integers  
	\begin{enumerate}
		\item infinitely many solutions $(x,y)$, if $\alpha = a/c,~ \beta=b/c$ with $a,b,c \in \N$, and $\gcd(a,c)|b$;
		\begin{align*}
		x &= c s - r_a, &
		y &= a s + r_c,
	\end{align*}
	with $s \geq S_0 := \max \left( \lceil r_a/c \rceil, \lceil -r_c /a \rceil \right)$, and $r_a$ and $r_c$ are integers such that $r_a a + r_c c = b$;
		\item exactly one solution $(x,y) = (q_1,q_2)$, if $\alpha \notin \Q$ and $\beta = q_2 - \alpha q_1 > 0$;
		\item no solution, otherwise. 
	\end{enumerate}	
\end{lemma}

\begin{proof}
	Let us start with rational slope $\alpha = a/c$, with $a,c \in \N$. In order to get integer solutions we need a rational $\beta = b/c$, with $b \in \N$. Then we need to find the solutions of the following linear Diophantine equation:
	\begin{align}
		\label{eq:barriersoldioph}
		cy - ax = b.
	\end{align}
	These solutions exist if and only if $\gcd(a,c) | b$. By the extended Euclidean algorithm we get integers $r_a, r_c \in \Z$ such that
	\begin{align*}
		r_a a + r_c c = b.
	\end{align*}	
	This is done by first computing numbers $r_a', r_c'$ such that $r_a' a/\gcd(a,c) + r_c' / \gcd(a,c) = 1$ and multiplying by $b$.
	All solutions are then given by the linear combination stated in the lemma. Due to the special form of \eqref{eq:barriersoldioph} with a positive and a negative coefficient in front of the unknowns, it follows that for all $s \geq S_0$ the solutions are positive. 
	
	Finally, let $\alpha$ be irrational. Assume there exist two points $Q = (q_1,q_2)$ and $P = (p_1, p_2)$ fulfilling the assumptions. By taking the difference we get $q_2-p_2 = \alpha (q_1-p_1)$ which implies that for $q_1 \neq q_2$ we get the contradiction $\alpha \in \Q$. But for $q_1 = q_2$ it also holds that $p_1 = p_2$ and therefore $Q=P$.

	It is easy to see that this solution exists if and only if $\beta = q_2 - \alpha q_1$ for arbitrary $q_1, q_2 \in \N$ as long as $\beta > 0$.	
\end{proof}

The previous lemma also appeared in~\cite{Kempner1917}, there, Kempner (of Kempner's series fame) also mentions that a similar claim 
holds for the number of algebraic rational (respectively algebraic) points on $y =\alpha x +\beta$ when  $\alpha$ is algebraic (respectively transcendental) slope. 
The lemma gives us all possible integer solutions on a boundary with rational slope. With this knowledge we can reformulate the lattice path integral from Theorem~\ref{theo:naka1} 
in order to give a more explicit result for all possible starting points and for any slope.

\begin{theo}[Lattice path integral and explicit binomial expression]
	\label{theo:nakaRational}
	Let $a,b,c$ be positive integers such that $\gcd(a,c)|b$. Let $r_a, r_c$ be integers such that $r_a a + r_c c = b$. Then, $q_1(s):=cs-r_a$ and $q_2(s):=as+r_c$ define all pairs $(q_1(s),q_2(s))$ of integers on the barrier~$L : y = \frac{a}{c} x + \frac{b}{c}$. Furthermore, let $V$ be the set of walks from $(q_1(s),q_2(s))$ to the origin strictly below the barrier~$L$.
	Then, we have
	\begin{align}
		\label{eq:nakarational}
		\int_{0}^{1} |W_t| \, dt &= \frac{b/c}{(a+c)s + (r_c - r_a)} \dbinom{(a+c)s + (r_c - r_a)}{a s +  r_c},
	\end{align}
	for $s \geq S_0:=\max \left( \lceil r_a/c \rceil, \lceil -r_c /a \rceil \right)$.
\end{theo}

For fixed $s$ the walks are ending after $q_1(s) + q_2(s) = (a+c)s + (r_c - r_a)$ steps, start at $(q_1(s),q_2(s))$ and go to the origin. In the equivalent formulation the walks start at $(q_1(s),q_2(s)-1)$ and go to the origin, but we consider $k = c\beta = b  $ different boundaries, given by

\begin{equation*}
	L_1: y < \frac{a}{c}x + \frac{b}{c}, \qquad
	L_2: y < \frac{a}{c}x + \frac{b-1}{c}, \text{\qquad \dots, \qquad }
	L_b: y < \frac{a}{c}x + \frac{1}{c}.	
\end{equation*}

\begin{example}
	Returning to Knuth's model we have $y < \frac{2}{5}x + \frac{2}{5}$. Thus, the explicit values are $a=b=2$ and $c=5$ and the assumptions of Theorem~\ref{theo:nakaRational} (Lattice path integral and explicit binomial expression) are  satisfied, as $\gcd(a,c)=1$. The Euclidean algorithm gives $r_a=-4$ and $r_c=2$. From Lemma~\ref{lem:rationalslope} on the possible starting point on the boundary, we deduce the possible integer coordinates on the barrier~$L$:
	\begin{align*}
		q_1(s) &= 5s + 4, &
		q_2(s) &= 2s + 2,
	\end{align*}
	for $s \geq 0$ which represent the starting points of the walks. Finally, Theorem~\ref{theo:nakaRational} directly gives the solution
	\begin{align*}
		\int_{0}^{1} |W_t| \, dt &= \frac{2/5}{7s+6} \dbinom{7s+6}{2s+2}.
	\end{align*}
	This value can be equivalently interpreted as the number of walks in $k=2$ models starting from $(5s+4,2s+1)$ and moving to the origin below the barriers
	\begin{align*}
		L_1: y &< \frac{2}{5}x + \frac{2}{5}, &
		L_2: y &< \frac{2}{5}x + \frac{1}{5}.
	\end{align*}
	This is exactly Knuth's problem, where his index $t=s+1$.
\end{example}

Formula~\eqref{eq:nakarational} directly yields nice lattice path identities in the manner of Knuth's problem. Yet, there are even more formulae of this type
that we will reveal in the next section. 
But let us start with an interesting (everyday) problem first.

\pagebreak
\section{Duchon's club and other slopes}
\label{sec:duchon}
\vspace{-2mm}
\subsection{Duchon's club: slope 2/3 and slope 3/2}

A Duchon walk is a Dyck path starting from $(0,0)$, with East and North steps,
and ending on the line $y=\frac{2}{3}x$ (see Figure~\ref{Duchonsclub}). This model 
was analysed by Duchon~\cite{Duchon00},
and further investigated by Banderier and Flajolet~\cite{BaFl02}, 
who called it the ``Duchon's club'' model,
as it can be seen as the number of possible ``histories'' of couples entering a club in the evening\footnote{Caveat: There are no real life facts/anecdotes hidden behind this pun!}, and exiting in groups of $3$.
What is the number of possible histories (knowing the club is closing empty)?
Well, this is exactly the number $E_n$ of excursions with $n$ steps $+2, -3$,
 or (by reversal of the time) the number of excursions with $n$ steps $-2, +3$.
 This gives the sequence $(E_{5n})_{n\in \mathbb{N}}=(1,2,23,377,7229,151491,3361598,\dots)$ (OEIS A060941).
In fact, these numbers $E_n$ appeared already in the article by Bizley~\cite{Bizley54} 
(who gave some binomial formulae, as we explained in Section~\ref{sec:imaginary}).
Duchon's club model should then be the Bizley--Duchon's club model;
Stigler's law of eponymy strikes again.

One open problem in the article~\cite{Duchon00} was the following one:
``The mean area is asymptotic to $K n^{3/2}$, but the constant $K$ can
only be approximated to $3.43$''.  Our method allows to identify this mysterious constant:

\vspace{-1mm}
\begin{theo}[Area below Duchon lattice paths]
The average area below Duchon excursions of length $n$ (lattice paths from $0$ to $0$, which jumps $-2$ and $+3$) 
is $$A_n \sim K n^{3/2} \text{ where } K=\sqrt{15 \pi}/2 \approx 3.432342124\,.$$
\end{theo}
\begin{proof}
The approach of~\cite{BaGi06} gives an expression for $A(z)=\sum A_n z^n$ in terms of the two small roots $u_1(z)$ and $u_2(z)$ of 
$1-z(1/u^2+u^3)=0$.
Then, using the rotation law gives the singular behaviour of $A(z)$, and therefore the asymptotics of $A_n$ with the explicit constant $K$.
\end{proof}

\vspace{-9mm}
\begin{figure}[!hb]
	\centering
	\scalebox{0.9}{
	\subfloat[North-East model: Dyck paths below the line of slope 2/3]{
		\includegraphics[width=0.37\textwidth]{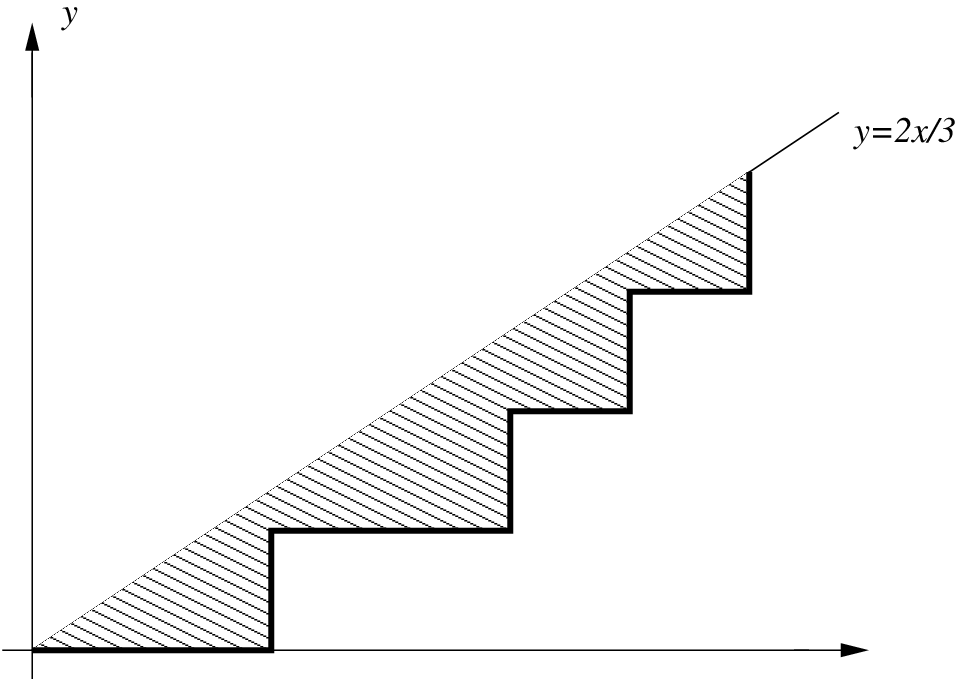}
	}
	\qquad
	\subfloat[Banderier--Flajolet model: excursions with $+2$ and $-3$ jumps]{
		\includegraphics[width=0.49\textwidth]{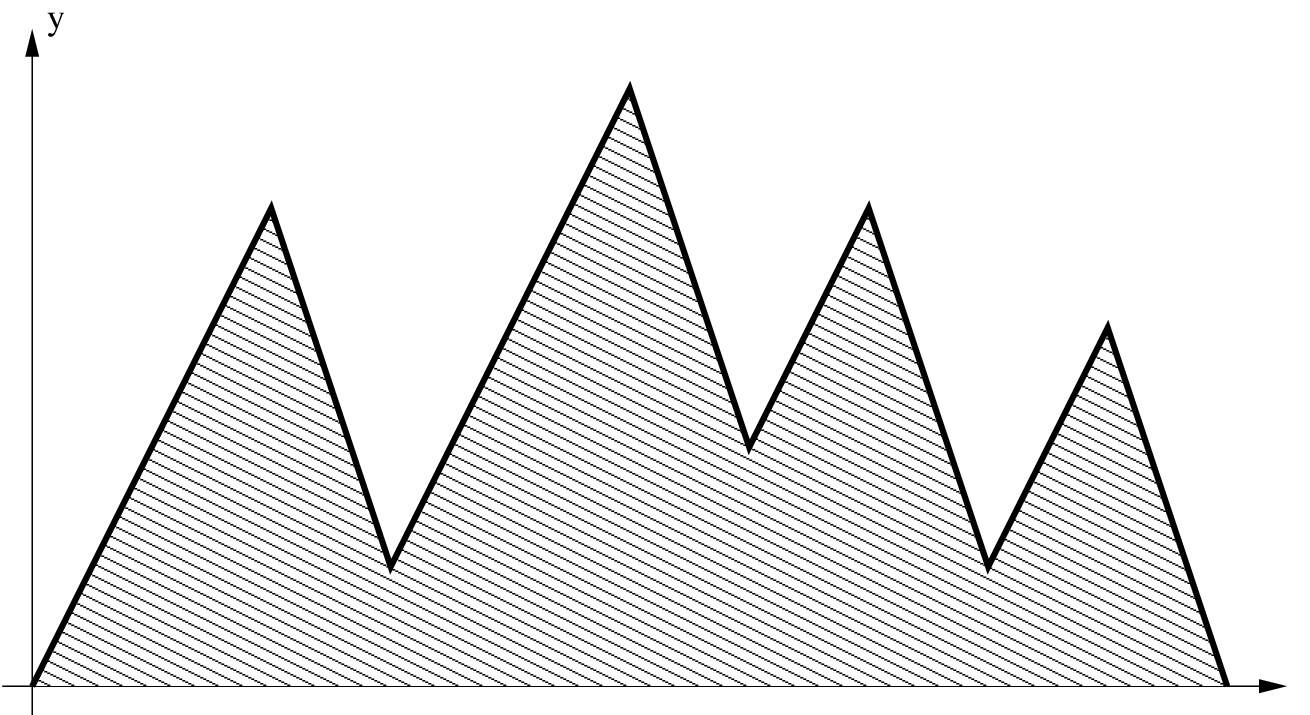}
	}}
	\caption{Dyck paths below the line of slope 2/3 and Duchon's club histories (i.e.,~excursions with jumps $+2, -3$) are in bijection.
Duchon conjectured that the average area  (in gray) after $n$ jumps is asymptotically equal to $K n^{3/2}$; 
our approach shows that $K=\sqrt{15 \pi}/2$. 
 }\label{Duchonsclub}
\end{figure}

\pagebreak
\subsection{Arbitrary rational slope}

The closed-form for the coefficient (Theorem~\ref{theo:closedformcoeff}) generalizes to arbitrary rational slope:

\begin{theo}[General closed-forms for any rational slope]
	Let $a,b,c$ be integers such that $\gcd(a,c)|b$. Let $A_s(k)$ be the number of Dyck walks below the line of slope $y=\frac{a}{c} x + \frac{k}{c}$, ending at $(x_s,y_s)$ given by
	\begin{align*}
		x_s &= c s - r_a, &
		y_s &= a s + r_c - 1,
	\end{align*}
	where $r_a$ and $r_c$ are integers such that $r_a a + r_c c = b$. These numbers are non-negative for $s \geq S_0 := \max \left( \lceil r_a/c \rceil, \lceil -r_c /a \rceil \right)$. Then it holds that
	\begin{align*}
		\sum_{k=1}^{b} A_s(k) &= \frac{b}{(a+c)s + (r_c - r_a)} \dbinom{(a+c)s + (r_c - r_a)}{a s +  r_c}.
	\end{align*}
\end{theo}

\begin{proof}
	This result is a direct consequence of Theorem~\ref{theo:nakaRational} (lattice path integral and explicit binomial expression) and the geometric bijection~\eqref{eq:nakaDiscrete}.
\end{proof}

The enumeration of lattice paths below the line $y=\frac{a}{c} x + \frac{b}{c}$ simplifies even more in the case~$a=b$. 
Additionally, we are able to extend the nice counting formula in terms of binomial coefficients.
In order to get these nice formulae, let us first state
what becomes the equivalent of Theorem~\ref{theo:closedformgf} (Closed-form for the generating function) in the case of any rational slope. 

\begin{lemma}[Schur polynomial closed-form for meanders ending at a given altitude]
	\label{lem:closedformgfgeneral}
	Let us consider walks in $\N^2$ with jumps $-a$ and $+c$ starting at altitude $h \geq a$. Let $u_1(z), \ldots,u_a(z)$ be the small roots of the kernel equation $1-zP(u)=0$, with $P(u) = u^{-a} + u^c$. Let $F_{0}(z), \ldots, F_{a-1}(z)$ be the generating functions of meanders ending at altitude $0, \ldots, a-1$, respectively. They are given by
\begin{align}
	\label{eq:Figeneral}
	F_i (z) &= \frac{(-1)^{a-i-1}}{z} s_{(h+1,1^{a-i-1},0^{i})} \left( u_1(z), \ldots, u_a(z) \right),
\end{align} 
	where $s_{\lambda}(x_1,\ldots,x_a)$ is a Schur polynomial in $a$ variables, and $\lambda = (\lambda_1,\ldots,\lambda_a)$ is an integer partition, i.e.,~$\lambda_1 \geq \lambda_2 \geq \cdots \geq \lambda_a \geq 0$. The notation $1^{s}$ denotes $s$ repetitions of $1$.
\end{lemma}

\begin{proof}
	Similar to \eqref{eq:funceq} for the given step set the functional equation is given by
	\begin{align*}
		(1 - z P(u)) F(z,u) = f_0(u) - z u^{-a} F_0(z) - z u^{-a+1} F_1(z) - \ldots - z u^{-1}F_{a-1}(z).
	\end{align*}
	Applying the kernel method, one may insert the $a$ small branches into this equation. Then one gets $a$ independent linear equations for the $a$ unknowns $F_0(z),\ldots, F_{a-1}(z)$. Expressing the solutions by Cramer's rule and rearranging the determinants, one uncovers the defining expressions for the claimed Schur polynomials (see e.g.~\cite[Chapter 7.15]{stan99} for an introduction to the relevant notions and notations).
\end{proof}

\begin{example}
	Let us consider the previous lemma for $a=3$. We get the linear system 
	\begin{align*}
		z
		\begin{pmatrix}
			1 & u_1(z) & u_1(z)^2 \\
			1 & u_2(z) & u_2(z)^2 \\
			1 & u_3(z) & u_3(z)^2 
		\end{pmatrix}
		\begin{pmatrix}
			F_0(z) \\
			F_1(z) \\
			F_2(z) 
		\end{pmatrix}
		=
		\begin{pmatrix}
			u_1(z)^{h+3} \\
			u_2(z)^{h+3} \\
			u_3(z)^{h+3} 
		\end{pmatrix}\,.
	\end{align*}
	Solving it with Cramer's rule and rearranging the determinants we get 
	\begin{align*}
		F_0(z) &= \frac{s_{(h+1,1,1)}(u_1,u_2,u_3)}{z}, &
		F_1(z) &= -\frac{s_{(h+1,1,0)}(u_1,u_2,u_3)}{z}, &
		F_2(z) &= \frac{s_{(h+1,0,0)}(u_1,u_2,u_3)}{z},
	\end{align*}
	by the definition of Schur polynomials.
\end{example}

Now, we are able to extend the results of the closed-form for the sum of coefficients (Theorem~\ref{theo:closedformcoeff}) even further.  At its heart lies the nice expression~\eqref{eq:AplusB}: $u_1^5+u_2^5$. 
We will see that such a phenomenon holds in full generality, 
involving a sum of $u_i^h$.

\begin{theo}[General closed-forms for lattice paths below a rational slope $y=\frac{a}{c} x + \frac{b}{c}$, with $b$ a multiple of $a$]
	\label{theo:genclosedform}
	Let $a,c$ be integers such that $a<c$, and let $b$ be a multiple of $a$. 
	Let $A_s(k)$ be the number of Dyck walks below the line of slope $y=\frac{a}{c} x + \frac{k}{c}$, $k \geq 1$, ending at $(x_s,y_s)$ given by
	\begin{align*}
		x_s &= c s - 1, &
		y_s &= a s - 1.
	\end{align*}
	Then it holds for $s \geq 1$ and $\ell \in \N$ such that $(\ell + 1) a < c$ that
	\begin{align*}
		\sum_{k=\ell a+1}^{(\ell+1)a} A_s(k) &= \frac{\ell a+c}{(a+c)s+\ell-1} \dbinom{(a+c)s + \ell - 1}{as-1}.
	\end{align*}
\end{theo}

\begin{proof}
	Consider walks starting at $(0,0)$, ending at $(x_s,y_s)$, and staying below the line $\frac{a}{c} x + \frac{1}{c}$. These are counted by $A_s(1)$. 
	Let us transform such walks by adding a new horizontal jump at the end. 
	Note that the first $\lfloor \frac{c}{a} \rfloor$ jumps must be horizontal jumps. 
	Thus, we can interpret this walk as one starting from $(1,0)$, ending at $(x_s+1,y_s)$ staying below the given boundary. But as a horizontal jump increases the distance to the boundary by $\frac{a}{c}$ this is equivalent to counting walks starting at $(0,0)$, ending at $(x_s,y_s)$, and staying below the boundary $\frac{a}{c} x + \frac{a+1}{c}$. This process is shown in Figure~\ref{fig:shiftingfirststeps}. Such walks are counted by $A_s(a+1)$.	
	
	\begin{figure}[h!t]
		\centering
		\includegraphics[height=50mm]{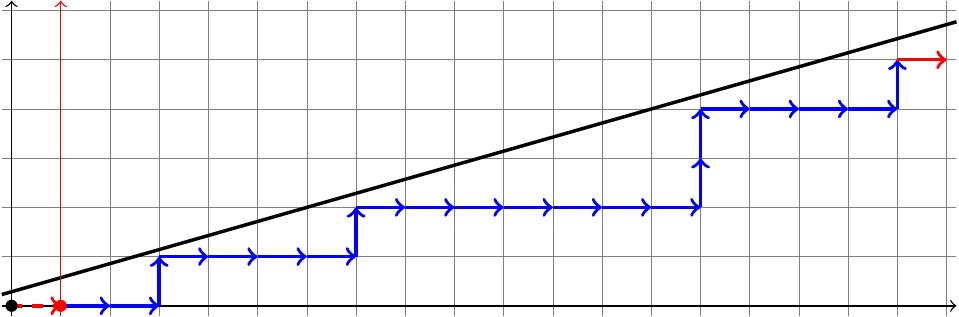}
		\caption{Transforming walks by moving the first step to the end of the walk. The red dot at $(1,0)$ and the red $y$-axis mark the new origin.}
		\label{fig:shiftingfirststeps}
	\end{figure}

	Thus, the sequence $A_s(1), A_s(a+1), A_s(2a+1), \ldots$ can be interpreted as counting walks staying always below the boundary $\frac{a}{c} x + \frac{1}{c}$, starting at $(0,0)$, and ending at $(x_s,y_s), (x_s+1,y_s), (x_s+2,y_s), \ldots$, respectively. In particular, for $\ell \geq 0$ we define these new ending points as $(\tilde{x}_s,\tilde{y}_s)$ given by
	\begin{align*}
		\tilde{x}_s &= x_s + \ell = c s + \ell - 1,&
		\tilde{y}_s & = y_s = a s - 1.
	\end{align*}
	Analogously, the same holds for $A_s(2), \ldots, A_s(a-1)$. 
	
		\smallskip
	For the start, we then follow the line of thought from Theorem~\ref{theo:closedformcoeff} (Closed-form for the sum of coefficients).
	Let us first derive the respective generating functions. 
	Therefore, we apply the bijection from Proposition~\ref{prop:bijgen}, reverse the time, and allow to touch $y=0$. Then the sum $\sum_{k=\ell a+1}^{(\ell+1)a} A_s(k)$ can be interpreted as walks of length $\tilde{x}_s + \tilde{y}_s = (a+c)s+\ell-2$, starting at altitude $a\tilde{x}_s- c \tilde{y}_s + i = \ell a + (c-a) + i$, and ending at altitude $i$ for $i=0,\ldots,a-1$. 	To simplify notation, let us introduce the constant 
	$$h:=\ell a + c\,.$$ 
 Then, walks end at $h - a + i$.
 	Therefore, we are now able to apply Lemma~\ref{lem:closedformgfgeneral} (Schur polynomial closed-form for meanders ending at a given altitude). Additionally, by reversing the summation order we get:
	\begin{align}
		\sum_{k=\ell a+1}^{(\ell+1)a} A_s(k) &= 
			[z^{(a+c)s+\ell-2}] \sum_{j=0}^{a-1} 
				\frac{(-1)^{j}}{z} s_{(h - j,1^{j},0^{a-j-1})} \left( u_1(z), \ldots, u_a(z) \right) \notag \\
				&= [z^{(a+c)s+\ell-1}] \left( \sum_{i=1}^a u_i(z)^{h} \right). 
				\label{eq:AplusBgeneral}
	\end{align}	
	This surprisingly simple result is due to a nice representation theorem of power symmetric functions in terms of Schur polynomials: \cite[Theorem~7.17.1]{stan99}. One gets this equation by setting $\mu=\emptyset$ and restricting the case to $a$ variables. Note that this is the analog of \eqref{eq:AplusB}. It is in one sense the reason for the nice closed-forms in this article.

	\smallskip	
	In contrast to Theorem~\ref{theo:closedformcoeff} (Closed-form for the sum of coefficients), we proceed now differently by Lagrange inversion~\cite{Lagrange70}. From the kernel method, we know that the small branches $u_i(z)$ satisfy the kernel equation $1-zP(u)=0$, where $P(u) = u^{-a} + u^c$ for general slope $a/c$. The entire form of the kernel equation satisfies nearly a Lagrangean scheme
	\begin{align*}
		u_i(z)^a &= z \left(1 + u_i(z)^{a+c} \right).
	\end{align*}
	By taking the $a$-th root, one gets for an auxiliary power series $U(x)$:
	\begin{align*}
		U(x) &= x \phi(U(x)), && \text{ with } & 
		\phi(u) &= \left(1 + u^{a+c} \right)^{1/a}.
	\end{align*}
	Let $\omega \neq 1$ be an $a$-th root of unity (i.e.,~$\omega^a=1$). Then we recover the $u_i(z)$, $i=1,\ldots,a$, by 
	\begin{align*}
		u_i(z) &= U\left( \omega^{i-1} z^{1/a} \right).
	\end{align*}
	Thus, coming back to \eqref{eq:AplusBgeneral} we are actually interested in 
	\begin{align*}
		\sum_{i=1}^a u_i(z)^{h} &= 
			\sum_{i=1}^a U\left(\omega^{i-1} z^{1/a}\right)^{h} = 
			\sum_{n \geq 0} U_{n} z^{n/a} \left(\sum_{i=1}^a \omega^{(i-1)n} \right) =
			a \sum_{n \geq 0} U_{a n} z^{n},
	\end{align*}
	where $U(x)^{h} = \sum_{n \geq 0} U_n x^n$ (in fact, by construction many coefficients $U_n$ are 0, because $U(z)$ has an $(a+c)$ periodic support, but this is not altering our reasoning hereafter). Considering~\eqref{eq:AplusBgeneral} again, we need $U_{an}$ for $n=(a+c)s+\ell-1$. 
	It is determined by the above Lagrangean scheme:
	\begin{align*}
		U_{an} &= [x^{a ((a+c)s+\ell-1)}] U(x)^h \\
		               &= \frac{\ell a + c}{a((a+c)s+\ell-1)} [u^{a ((a+c)s+\ell-1) - 1}] u^{\ell a + c -1} \left(1 + u^{a+c}\right)^{(a+c)s+\ell-1} \\
							&= \frac{\ell a + c}{a((a+c)s+\ell-1)} \binom{(a+c)s + \ell -1}{as-1}.
	\end{align*}
	Rewriting the binomial coefficient by symmetry, the claim follows.
\end{proof}

\begin{example}
	Knuth's original problem was dealing with boundaries $y = \frac{2}{5}{x} + \frac{k}{5}$, $(k=1,\ldots,4)$. In particular, we may choose $\ell = 0$, and $\ell = 1$ to get:
	\begin{align*}
		\sum_{k=1}^{2} A_{s}(k) &= \frac{5}{7s-1} \binom{7s-1}{2s-1} = \frac{2}{7s-1} \binom{7s-1}{2s},\\
		\sum_{k=3}^{4} A_{s}(k) &= \frac{1}{s} \binom{7s}{2s-1}.
	\end{align*}
	The first one is the known result, whereas the second one is yet another surprising identity.
\end{example}

Now, we come back to the asymptotics of Section~\ref{sec:asymptotics}. Some  key ingredients were Proposition~\ref{prop:periodicasympt} (Periodic rule of thumb) and the rotation law of the small branches. Happily, such a rotation law holds in general for any slope, and the derived techniques can also be applied. This is what we present now.

Let $P(u) = u^{-a} + u^c$ be the jump polynomial of directed walks. Thus, we have $a$ small branches $u_i(z)$ satisfying the kernel equation $1-zP(u_i(z)) = 0$. As before let $\tau$ be the unique positive root of $P'(\tau)$, and let $\rho$ be defined as $\rho = 1/P(\tau)$. 
Recall that the small branches are possibly singular only at the roots of $P'(u)$. 
The jump polynomial has periodic support with period $p=a+c$ as $P(u) = u^{-a} H(u^p)$ with $H(u) = 1 + u$. Hence, there are $p$ possible singularities of the small branches
\begin{align*}
	\zeta_k &= \rho \omega^k, \qquad \text{ with } \qquad \omega = e^{2 \pi i / p}.
\end{align*}

\pagebreak
The general version of Lemma~\ref{lem:u1u2} reads then as follows:

\begin{lemma}[Rotation law of small branches]
	\label{lem:uirotation}	
	Let $\gcd(a,c)=1$. Then there exists a permutation~$\sigma$ of $\{1,\ldots,p\}$ without fix points and an integer $\kappa$ (satisfying   $\kappa a +1  \equiv  0 \mod p$) such that 
	\begin{align*}
		u_i(\omega z) &= \omega^{\kappa} u_{\sigma(i)} (z),
	\end{align*}
	for all $z \in \C$ with $|z| \leq \rho$ and $0 < \arg(z) < \pi-2\pi/p$.
\end{lemma}

\begin{proof}
		We proceed as in the proof of Lemma~\ref{lem:u1u2}. Define $U(z) := \omega^\kappa u_i(\omega z)$ and a function $X(z) := U^a - z \phi(U)$ with $\phi(u) := u^a P(u)$. Then a straightforward computation shows that 
		\begin{align*}
			X(z) &= \left( \omega^\kappa u_i(\omega z) \right)^a - z \phi \left( \omega^\kappa u_i(\omega z) \right) 
			     = \omega^{\kappa a} u_i(\omega z)^a - z \phi(u_i(\omega z)),
		\end{align*}
		as $\phi(u)$ is $p$-periodic. Therefore, we get by the following transformation
		\begin{align*}
			\omega X(z/\omega) &=  \omega^{\kappa a+1} u_i(z)^a - z \phi(u_i(z)) = 0,
		\end{align*}
		if $\kappa a+1 \equiv 0 \mod p$, because of the kernel equation. Thus, $X = U^a - z\phi(U)=0$ and therefore $U(z)$ is a root of the kernel equation. It has to be a small root, as it is converging to $0$ if $z$ goes to $0$.
		Furthermore, it has to be a different root, as it has a different Puiseux expansion. By the analytic continuation principle (as long as we avoid the cut line $\arg(z)=-\pi$) the result follows.
\end{proof}

The last lemma allows us to state the following ``meta''-result:

\begin{theo}[Metatheorem/rule of thumb: enumeration and asymptotics of lattice paths]
Constrained lattice paths have an algebraic generating function, expressible in terms of Schur functions (a symmetric function involving the small branches of the kernel).
Singularity analysis gives its asymptotic behaviour, 
which is equal to the asymptotics at the dominant real singularity (times the periodicity whenever the rotation law holds).
\end{theo}

We call this a metatheorem because it is rather informal in the description of the constraints allowed (it could be positivity, prescribed starting or ending points, 
to live in a cone, to stay below a line of rational slope, to have some additional Markovian behaviour, to be multidimensional with one border, or in bijection with any of these constraints...), 
 in all these cases the spirit of the kernel method and analytic combinatorics should give the enumeration and the asymptotics.
Different incarnations of this rule of thumb appear in~\cite{BanderierDrmota,BaFl02,BaGi06,bani10,BousquetMelou08}, and no doubt that many new lattice problems 
on the one hand, and many new combinatorial problems involving some type of periodicity on the other hand,
will offer additional incarnations of this metatheorem.

\begin{figure}[!ht] \begin{center}
\begin{tabular}{cc}
\begin{tabular}{c}
\includegraphics[width=5cm]{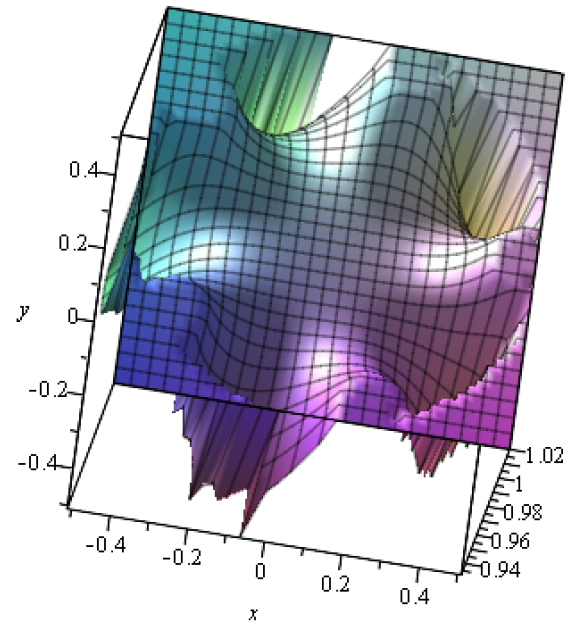}
\end{tabular}
&\begin{tabular}{c}
\begin{minipage}{0.5\textwidth}
This is the landscape in the complex plane of $|F(z)|$, where $F$ is here the generating function of Duchon's club excursions. One can see the five dominant singularities. It is enough to know the local behaviour 
near the real positive singularity, the rotation law implies the same behaviour at the other dominant singularities.\end{minipage}
\end{tabular}
\end{tabular}
\end{center}
\vspace{-1\baselineskip}
\caption{Landscape in the complex plane of the generating function of lattice paths.}
\label{duchon3d}
\end{figure}

 \pagebreak
\section{Conclusion}
\label{sec:conclusion}
In this article, we analysed some models of directed lattice paths below a line of rational slope.
As a guiding thread, we first illustrated our method on Dyck paths below the line of slope $2/5$.
Beside the (pleasant) satisfaction of answering a problem of Don Knuth,
this sheds light on properties of constrained lattice paths, including the delicate case (for analysis) of a periodic  
behaviour.

	We can shortly recall the main methods used in this article to attack lattice path problems:
	
	Firstly, the method of choice of Nakamigawa and Tokushige was the \emph{cycle lemma}. 
	It is a classical result for lattice paths which uses the geometry of the problem. However, its applications are limited to certain cases. 
	
	Secondly, a more general result is given in Theorem~\ref{theo:genclosedform} (General closed-forms for lattice paths below a rational slope $y=\frac{a}{c} x + \frac{b}{c}$),
	 via the \emph{Lagrange inversion}.  This directly gives the sought closed-form. 
	 However, it does not give access to the asymptotics.
	
	Thus, thirdly, we used the \emph{kernel method} to express the generating functions explicitly in terms of (known) algebraic functions. 
	This gave us access to the asymptotics, and is an alternative way to access the closed-forms.
	Our Proposition{~\ref{prop:periodicasympt} (Periodic  rule of thumb) explains in which way  the asymptotic expansions are modified in the case of a periodic behaviour
	(via some local asymptotics extractor and the rotation law);  we expect this approach to be reused in many other problems.
	
Also, the method of \emph{holonomy theory} used in Theorem~\ref{theo:closedformcoeff} (Closed-form for the sum of coefficients) shows the possible usage of computer algebra to {\em prove} such {\em conjectured} identities. This is probably the fastest technique for checking given identities, and can be automatized to a great extent.  The interested reader is referred to the nicely written introductions \cite{pewz96, KauersPaule11}.

Our approach extends to any lattice path (with any set of jumps of positive coordinates) below a line of (ir)rational slope  (see~\cite{BanderierWallner16}).
This leads to some nice universal results for the enumeration and asymptotics.
As an open question, it could be natural to look for similar results for lattice paths (with any set of jumps with positive and negative coordinates, and not just jumps to the nearest neighbours) in a cone given by two lines of rational slope.
This is equivalent to the enumeration of non-directed lattice paths in dimension $2$. 
Despite the nice approach from the probabilist 
school~\cite{Fayolle99,DenisovWachtel15} and from the combinatorial 
school~\cite{BousquetMelouMishna10} via the iterated kernel method, this remains a terribly simple problem (to state!),
but a challenge for the mathematics of this century.

\medskip

\textbf{Acknowledgments:}
\label{sec:ack}
This work is the result of a collaboration founded by the SFB project F50 ``Algorithmic and Enumerative Combinatorics'' 
and the Franco-Austrian PHC ``Amadeus''. Michael Wallner is supported by the Austrian Science Fund (FWF) grant SFB F50-03 and by \"OAD, grant F04/2012. 
A preliminary version of this work~\cite{BanderierWallner15} was presented at the conference ANALCO'15 (San Diego, January 2015) and at the 
8th International Conference on Lattice Path Combinatorics \& Applications (Pomona, August 2015).
Last but not least, we thank Don Knuth, Ernst Schulte-Geers, and Manuel Kauers for exchanging references on this problem, 
and the referee for the detailed feedback!

\addcontentsline{toc}{chapter}{References}
\bibliographystyle{plain}
\bibliography{knuthslopearxiv}
\label{sec:biblio}

\end{document}